\documentclass[12pt]{article}

\usepackage[margin=1in]{geometry}
\usepackage[T1]{fontenc}
\usepackage[utf8]{inputenc}
\usepackage{lmodern}
\usepackage{amsmath,amssymb,amsthm}
\usepackage{hyperref}
\usepackage{graphicx}
\usepackage{enumitem}
\usepackage{xcolor}
\usepackage{setspace}
\usepackage{titlesec}
\usepackage{placeins}

\hypersetup{
 colorlinks=true,
 linkcolor=blue,
 urlcolor=blue,
 citecolor=blue
}

\titleformat{\section}{\large\bfseries}{\thesection.}{0.5em}{}
\titleformat{\subsection}{\normalsize\bfseries}{\thesubsection.}{0.5em}{}

\newtheorem{theorem}{Theorem}[section]
\newtheorem*{theoremrestated}{Theorem~\ref{thm:pairlaw}, restated}

\newtheorem{corollary}{Corollary}

\newtheorem {conjecture}{Conjecture}

\title{The Fully Depolarizing Noise Conjecture for Entangled Physical States: A Twenty-Year Perspective}
\author{Gil Kalai}

\begin {document}

\maketitle

\begin{abstract}
In this paper I revisit my 2006 conjecture on correlated errors in entangled
physical qubits, originally proposed as a potential obstruction to quantum
fault tolerance. The conjecture asserts that, in any physical implementation
of a quantum computer, the effective noise channel acting on entangled physical
qubits contains a joint fully depolarizing component, with a rate comparable
to that of two-qubit gate errors. This hypothesized structural constraint goes
beyond standard noise models and, if valid, would pose a significant challenge
to scalable quantum fault tolerance.

The conjecture remains open, but recent advances in experimental quantum
computing bring it within reach of empirical testing on current devices. I also
discuss two related directions in my critical study of quantum computation: the
role of noise sensitivity and computational complexity in noisy
intermediate-scale quantum systems, and the statistical analysis of experimental
claims of quantum advantage. Finally, since this paper is written for a volume
honoring Yuri Gurevich, I include some reflections on the ways in which my
scientific and personal trajectory became intertwined with Yuri's.
\end{abstract}

\section{Introduction}

I have long held the view that computationally superior quantum computers, and
even early milestones toward this goal, are inherently impossible. Moreover, I
expect the inherent obstacles to quantum computation to be manifested already
in quantum devices with a small number of qubits. In this paper I revisit a
conjecture from 2006 about the nature of noise in entangled physical states,
and I discuss this conjecture, my subsequent work, and my skepticism regarding
quantum computers from a broader perspective.

A model of computation in the physical world must come with an error model. It
is known, and rather easy to prove, that if engineers could reduce the error
per operation to an arbitrarily small level, roughly of order one over the
number of gates in the computation, then noisy quantum circuits could
approximate ideal noiseless circuits well enough for algorithmic purposes.
However, most physicists and computer scientists regard the assumption that
the noise level can be reduced to arbitrarily small values, even values
polynomially small in the number of qubits, as physically unrealistic.

The threshold theorem \cite {AhaBen97,KLZ98,Kit97}  is a landmark in the theory of quantum computation. It
asserts, roughly speaking, that fault-tolerant quantum computation is possible
provided that the noise rate is below a certain positive constant and that the
noise satisfies suitable locality and weak-correlation assumptions. Since the
late 1990s, several researchers have questioned whether these assumptions
accurately describe physical quantum systems. Some argued that the independence
or weak-correlation assumptions are unrealistic; others argued that the
required noise levels are far beyond what physics will allow.

My own work between 2005 and 2013 studied models of correlated errors. The
research question I asked was: what kind of law for correlated errors could
cause quantum error correction to fail? My proposal was to formulate a
constraint on the error model according to which errors on entangled physical
qubits are positively correlated.

A fully depolarizing error on a pair of qubits replaces their joint state by
the completely random two-qubit state, thereby erasing the quantum information
stored in the pair and its correlations with the rest of the system. Such two-qubit fully depolarizing errors are among the standard local noise
components used in formulations and simulations of quantum fault tolerance, and they are also
a familiar effective component of the noise observed in elementary two-qubit
entangling gates in current devices.

\begin{conjecture}[The Fully Depolarizing Noise Conjecture, informal]
Whenever two physical qubits are entangled, the effective noise acting on them
contains a joint fully depolarizing component. The weight of this component is
bounded below by a positive quantity depending on the physical implementation
and on the amount of entanglement between the qubits. In particular, for Bell
pairs, this weight is expected to be comparable to the rate of fully
depolarizing noise in elementary two-qubit entangling gates.
\end{conjecture}

The novelty of the conjecture is not the appearance of two-qubit fully depolarizing noise itself, but the assertion that a comparable component necessarily persists for entangled physical qubits, including those whose entanglement is generated indirectly, with a strength comparable to that of elementary entangling gates.

In this paper I revisit this conjecture from a twenty-year perspective and
discuss, more briefly, other aspects of my work on quantum computation.

In Section~\ref{s:bell} I formulate the conjecture first for Bell states. In
Section~\ref{s:fdnc} I formulate it for general entangled states, and more
generally for pairs of qubits with localizable entanglement, in the sense of
Verstraete--Popp--Cirac \cite{VPC04}. I then describe implications for
entangled states on more than two qubits. In Section~\ref{s:ce} I discuss
some consequences of the conjecture for multi-qubit entangled states, including
states arising in quantum error-correcting codes. One notable consequence is
noise synchronization; another is a discrepancy between error rates measured
in trace distance and error rates measured by qubit-error counts. 
In Section~\ref{s:exp} I discuss ways to test the conjecture experimentally.
Such tests appear feasible already on current quantum devices, although they
also present substantial experimental and statistical challenges. 

In Section~\ref{s:nscc} I briefly describe my subsequent work on computational
complexity and noise sensitivity for noisy intermediate-scale quantum systems.
In Section~\ref{s:stat} I describe statistical tools for examining experimental
claims about NISQ (noisy intermediate-scale
quantum) computers.\footnote {Scrutinizing both experimental and theoretical claims in quantum computing is especially important since billion-dollar industries and large public investments are involved.} 
Section~\ref{s:phys} discusses some connections to
physics and relates them to concerns about quantum computers raised by Alicki, and others.

Since this paper is written for a volume honoring Yuri Gurevich, I also include
a personal section. In Section~\ref{s:yuri} I describe my friendship with Yuri
and several broader themes on which our interests have intersected; these
themes also provide wider perspectives on various aspects of my
quantum-computation research. Section~\ref{s:conc} concludes. Finally,
Appendix~\ref{s:counter} presents several counterarguments to the conjecture
and my responses to them, Appendix~\ref{s:moment} contains a simple
probabilistic theorem illustrating how pairwise correlated errors can force
large synchronized error events, and Appendix~\ref{s:back} describes background and related literature regarding quantum fault tolerance and correlated noise and regarding quantum advantage. 

\section{The Bell-State Form of the Fully Depolarizing Noise Conjecture}
\label{s:bell}

\begin{quote}
\emph{\textbf{Conjecture (Fully Depolarizing Noise for Bell States).}}\\
\emph{For every physical implementation of a quantum computer, whenever a Bell
state on a pair of physical qubits is created, the effective noise channel
contains a joint fully depolarizing component: with some small probability
$t>0$, both qubits are replaced by the maximally mixed state. }
\end{quote}

Equivalently, the preparation of a Bell state, or two-qubit cat state, on two
physical qubits necessarily carries a non-negligible probability that
\emph{both} qubits undergo fully depolarizing noise. I further conjectured that $t$ is 
comparable to the rate of fully
depolarizing noise in elementary two-qubit entangling gates.

Twenty years ago I proposed that this phenomenon cannot be avoided by any
method of preparing a Bell state on a pair of physical qubits. The conjecture
was stated more broadly: it applied to every entangled pair of physical qubits,
not necessarily maximally entangled ones. Moreover, the weight of the joint
fully depolarizing component was conjectured to grow at least linearly with the
amount of entanglement between the qubits. Even more generally, it applies to
pairs of qubits with localizable entanglement in the sense of
Verstraete--Popp--Cirac \cite{VPC04}. These more general formulations are
discussed in Section~\ref{s:fdnc}.

To my knowledge, the conjecture remains open. I will refer to it as FDNC.

\subsection{Bell states and entangled states}

A \emph{Bell state} is a maximally entangled state of two qubits. A canonical
example is
\[
\frac{1}{\sqrt{2}}(|00\rangle+|11\rangle).
\]
The four standard Bell states form an orthonormal basis:
\[
\frac{1}{\sqrt{2}}(|00\rangle\pm|11\rangle),
\qquad
\frac{1}{\sqrt{2}}(|01\rangle\pm|10\rangle).
\]
They are obtained from one another by local unitary operations.

Bell states represent the simplest and most fundamental form of quantum
entanglement. They play a central role in quantum information theory and are
the basic resource behind protocols such as quantum teleportation, entanglement
swapping, and many constructions in quantum error correction.

More generally, a pure state of two qubits can be written as
\[
|\psi\rangle =
a|00\rangle+b|01\rangle+c|10\rangle+d|11\rangle,
\]
where $a,b,c,d\in\mathbb C$ and
\[
|a|^2+|b|^2+|c|^2+|d|^2=1.
\]
Such a state is called \emph{entangled} if it cannot be written as a tensor
product
\[
|\psi\rangle=|\alpha\rangle\otimes|\beta\rangle
\]
of single-qubit states. Equivalently, for a pure two-qubit state,
$|\psi\rangle$ is entangled if the reduced density matrix of either qubit is
mixed.

Every two-qubit pure state can be written, after suitable local unitary
transformations on the two qubits, in the form
\[
|\psi\rangle
=
\sqrt{\lambda}\,|00\rangle
+
\sqrt{1-\lambda}\,|11\rangle,
\]
where $0\le \lambda\le 1$. This representation is the Schmidt decomposition.
The state is entangled whenever $0<\lambda<1$, and it is a Bell state when
$\lambda=1/2$.

Up to local unitary transformations, two-qubit pure states are described
by this one-parameter family, with maximal entanglement at \(\lambda=1/2\).

\paragraph{A remark on Bell-state symmetry.}
For an isolated Bell state on two qubits, full depolarization of one qubit and
joint full depolarization of the two-qubit pair give the same output density
matrix, namely \(I_{AB}/4\). Thus Bell-state tomography alone cannot distinguish
these two error mechanisms. It is one reason why 
channel tomography, or larger states with localizable entanglement are
needed for a clean experimental test of the conjecture.

The conjecture concerns effective noise channels, not merely their action on a
single two-qubit state. This point becomes important in larger systems. If a
Bell pair is obtained by localizing entanglement from a many-qubit state, then
one-qubit depolarization and joint depolarization of the localized pair need not have the same effect on the full system before the localization
measurement. Likewise, in a state where many pairs have localizable
entanglement, requiring joint depolarizing components for many pairs is not
equivalent to ordinary single-qubit depolarization. Thus the conjecture should
be formulated at the level of effective noise channels, and not only at the
level of the final two-qubit density matrix of an isolated Bell pair.


\subsection{Direct versus indirect creation of entanglement}

For directly gated qubits, a joint two-qubit depolarizing component is already
part of standard stochastic noise models. The novelty of the conjecture is the
claim that a comparable error component remains present even when the Bell
state is created indirectly. Quantitatively, the conjectured value of $t$ is
in the ballpark of the fully depolarizing noise rate for elementary two-qubit
entangling gates.

Here is a simple indirect construction. Start with qubit 1 in the state
\[
\frac{1}{\sqrt2}(|0\rangle+|1\rangle),
\]
and qubits 2 and 3 in the state $|0\rangle$. Apply a CNOT with control \(1\) and target \(2\), followed by a CNOT
with control \(2\) and target \(3\). This produces the three-qubit GHZ state
\[
\frac{1}{\sqrt2}(|000\rangle+|111\rangle).
\]
Measuring qubit 2 in the $X$-basis then produces a Bell state on qubits 1 and
3, up to a Pauli correction depending on the measurement outcome.

Now consider a simplified independent gate-level noise model in which each
two-qubit gate independently fully depolarizes its participating pair with
probability \(t\), and all other operations are ideal. First count only
the endpoints directly acted on by these elementary depolarizing faults:
\begin{itemize}[leftmargin=2em]
 \item qubit \(1\) is directly depolarized by a fault on the first gate,
       with probability \(t\);
 \item qubit \(3\) is directly depolarized by a fault on the second gate,
       with probability \(t\);
 \item both endpoints are directly acted on by elementary depolarizing
       faults only when both gates fault, with probability \(t^2\).
\end{itemize}
This counts the locations of the elementary faults, not their full effects
after propagation through the circuit.

\paragraph{Remark on the effective channel.}
To state the relevant claim precisely, regard the entire mediated procedure
as a two-qubit gate acting on arbitrary inputs to qubits \(1,3\), with
qubit \(2\) initialized in \(|0\rangle\). Retain both outcomes of the
mediator's \(X\)-measurement and apply \(Z_1\) for outcome \(-\).
The ideal endpoint gate is a CNOT from qubit \(1\) to qubit \(3\).
For the noise model above, a calculation of the effective endpoint channel
shows that its largest joint fully depolarizing component has weight
exactly \(t^2\).\footnote{The equality can be checked as follows. The branch in which both gates
fault has probability \(t^2\) and fully depolarizes the endpoints,
giving the lower bound. For the input \(|0+\rangle\), the no-fault,
first-fault-only, and second-fault-only branches give, respectively,
\(|0+\rangle\langle0+|\),
\(I_1/2\otimes|+\rangle\langle+|_3\), and
\(|0\rangle\langle0|_1\otimes I_3/2\).
Thus measuring \(Z_1,X_3\) gives the outcome \((1,-)\) with probability
\(t^2/4\). Any joint fully depolarizing component of weight \(w\)
would contribute at least \(w/4\) to this probability, so \(w\le t^2\).}

This is a channel-level statement, not a statement about the total noise
in the final Bell state. A single elementary fault can already produce
the maximally mixed output \(I_{13}/4\) for the Bell-preparation input,
without acting as full two-qubit depolarization on arbitrary inputs.
Indeed, the Bell-state output has a maximally mixed contribution of weight
\(2t-t^2\). Thus the output noise remains of order \(t\), while the joint
fully depolarizing channel weight is only \(t^2\).

\paragraph{The quantitative prediction of FDNC.}
The conjecture asserts that nature {\em does not permit} this quadratic
suppression of the channel-level weight. Even when entanglement is
generated indirectly---through mediating qubits, measurements, feedforward,
or circuit identities---the effective noise channel should retain a joint
fully depolarizing component of order \(t\), comparable to the elementary
two-qubit noise scale.

For illustration, if the elementary joint fully depolarizing weight is
\(t=5\times10^{-3}\), then the simplified model gives
\(t^2=2.5\times10^{-5}\) for the mediated gate. FDNC predicts a weight
comparable to the elementary scale instead. Here \(t\) denotes a
specified channel-component weight, not an arbitrary measure of total
gate infidelity.

\subsection{Physical intuition}

The conjecture expresses, in a mathematically sharp way, a common physical
intuition:

\begin{quote}
Entanglement between two physical systems requires a genuine physical
interaction, and such interaction inevitably exposes both systems to correlated
noise.
\end{quote}

For example, one should not expect an indirect method of producing an
entangling operation between two transmons to be much more reliable, at the
level of joint fully depolarizing noise, than a direct physical interaction
between them. Standard gate-level noise modeling does not enforce this. Indeed,
in the simple circuit model described above, indirect constructions reduce the
joint fully depolarizing error from order $t$ to order $t^2$. Thus the
conjecture postulates an additional and more restrictive structural property of
physical noise, beyond standard local stochastic models.

\subsection{Further discussion of the conjecture}

We record here several further aspects of the conjecture. Its more general
formulation for entangled pairs is given in Section~\ref{s:fdnc}, and its
consequences for entangled states on more than two qubits are discussed in
Section~\ref{s:ce}.

\paragraph{The state of the conjecture.}
To my knowledge, the conjecture remains open. Current NISQ devices could
in principle test it, even at noise levels above the fault-tolerance threshold.
The main difficulty is not the number of qubits, but the identification of the
specific correlated fully depolarizing component in the effective noise
channel. We return to experimental tests in Section~\ref{s:exp}.

\paragraph{The conjecture is structural and not mechanistic.}
Some objections to the conjecture treat it as if it were proposing a specific
mechanistic claim of the following form: here is the physical source of the
noise---for example microscopic Hamiltonian couplings, thermal photons,
leakage, or crosstalk---and here is the dynamical derivation showing why fully
depolarizing correlations appear. Such a claim would specify the environment,
the interaction model, the time evolution, and the exact channel arising from
tracing out the bath.

My conjecture is not mechanistic in this sense. It is a structural claim about
the form of the effective noise channel:

\begin{quote}
\emph{Whenever two physical qubits can be prepared, or projected, into an
entangled state, there is an intrinsic order-$t$ fully depolarizing component
acting jointly on them.}
\end{quote}

Here $t$ is a physical noise scale, comparable to
the rate of fully depolarizing noise in  two-qubit entangling gates.
The conjecture concerns the form of the effective channel, not the microscopic
process generating it. Of course, mechanistic explanations, which may differ
between physical implementations, would be valuable if they could be found.

\paragraph{Gated qubits and ``purifying'' gate errors.}
For directly gated qubits, a joint two-qubit depolarizing component is already
part of standard stochastic noise models. The novelty of the conjecture is the
claim that an entangled pair prepared indirectly cannot avoid a comparable
joint fully depolarizing component. In ordinary independent gate-level models,
one may start with two-qubit gates that include fully depolarizing noise at
rate $t$, and nevertheless arrange that the final indirectly prepared pair has
a joint fully depolarizing component only of order $t^2$. In this sense, the
model appears to ``purify'' the physical two-qubit gate noise. The conjecture
asserts that such purification is not physically possible for entangled
physical qubits.

We note that fault-tolerance schemes allow a very strong form of purification. For error channels based on the standard noise models, and for complicated entangled states on $n$ physical qubits, quantum fault tolerance allows exponential suppression of accumulated errors. In particular, it predicts exponentially small effective joint depolarizing errors for most pairs of physical qubits whose recent causal histories do not meet.

\subsection{Motivation and sources}

The original motivation for the conjecture was to ``reverse engineer'' natural
structural conditions on noise that would cause quantum fault tolerance to
fail. For this reason, I did not view the conjecture itself as a reason for
people to revise their \emph{a priori} beliefs about quantum computers. Rather,
I regarded it as a concrete and testable benchmark for quantum devices---one
that is meaningful both for small systems with only a handful of qubits and for
larger systems. The conjecture is relevant even to systems operating at noise
levels above the threshold required for fault tolerance.

An early version of the conjecture 
appeared as {\it Postulate 1} in my 2006 paper \cite{Kal06},
\emph{How quantum computers can fail}, and in several subsequent works.\footnote {In some early papers I also formulated related conjectures for correlated
\emph{classical} noisy systems, asserting that correlation for the ``signal''
implies correlation for the noise. In this generality, classical computation
provides simple counterexamples. Nevertheless, suitably adjusted versions of
the conjecture appear to apply to several natural classical physical systems.} 
The
related ``noise synchronization'' consequence for highly entangled states
appeared as {\it Postulate 2}, and it connects back to my 2005 paper
\cite{Kal05}. These two conjectures together with a further consequence
for the error rate measured in terms of qubit errors, became Predictions 1, 2,
and 3 in my 2016 paper \cite{Kal16}, \emph{The quantum computer puzzle}.
These conjectures also played a central role in my 2012 debate with Aram Harrow \cite{KalHar12}.

\section{Formulation of the Conjecture for General Entangled States}
\label {s:fdnc}
In this section we formulate the Fully Depolarizing Noise Conjecture
for general entangled states using entropy of entanglement.

\subsection{Entropy of entanglement}

Let $A$ and $B$ be two qubits and let $|\psi\rangle_{AB}$ be a pure state on
$\mathbb{C}^2\otimes\mathbb{C}^2$.
Let
\[
\rho_A = \operatorname{Tr}_B(|\psi\rangle\langle\psi|)
\]
be the reduced density matrix on $A$.
The \emph{entropy of entanglement} of $|\psi\rangle$ is defined by
\[
E(\psi) := S(\rho_A),
\]
where $S(\rho)=-\operatorname{Tr}(\rho\log_2\rho)$ denotes the
von Neumann entropy.

The entropy of entanglement satisfies
\[
0 \le E(\psi) \le 1.
\]
It vanishes for product states and equals $1$ for maximally entangled
(two–qubit Bell) states.

\subsection{Joint fully depolarizing channel}

Let $A$ and $B$ be two physical qubits in a larger quantum system.
The \emph{joint fully depolarizing channel} acting on $A,B$
is the completely positive trace-preserving map
\[
\mathcal D_{AB}(\rho)
= \frac{I_{AB}}{4}\otimes \operatorname{Tr}_{AB}(\rho),
\]
where $I_{AB}$ is the identity on the two-qubit space.

This channel replaces the state of qubits $A,B$
by the maximally mixed state and removes all correlations
between $A,B$ and the rest of the system.

\subsection{Depolarizing weight of a noise channel}

Let $\mathcal E$ be a completely positive trace-preserving map describing
the effective noise channel acting on the full quantum system.
We define the \emph{depolarizing weight} of $\mathcal E$ on the pair
$(A,B)$ by
\[
w_{AB}(\mathcal E)
=
\sup
\left\{
\lambda\in[0,1]:
\mathcal E
=
\lambda\,\mathcal D_{AB}
+
(1-\lambda)\,\mathcal F
\text{ for some CPTP map }\mathcal F
\right\}.
\]

Thus $w_{AB}(\mathcal E)$ measures the largest possible coefficient of the
joint fully depolarizing channel $\mathcal D_{AB}$ that can appear in a
convex decomposition of the noise channel $\mathcal E$.

\subsection{The Fully Depolarizing Noise Conjecture}

We now formulate the conjecture in terms of entropy of entanglement.

\medskip

\noindent
\textbf{Conjecture (Fully Depolarizing Noise Conjecture).}
\emph{
There exists a constant $c>0$, depending only on the physical
implementation of the device, such that whenever two physical qubits
$A,B$ are ideally prepared in a pure state $|\psi\rangle_{AB}$,
the effective noise channel $\mathcal E$ acting on the system satisfies
\[
w_{AB}(\mathcal E) \ge c\,E(\psi),
\]
where $E(\psi)$ is the entropy of entanglement of $|\psi\rangle$. 
}

\medskip

In particular, for maximally entangled states $E(\psi)=1$,
and the conjecture asserts that the noise channel must contain
a fully depolarizing component of weight at least $c$. Also here, I conjecture that 
$c$ is comparable to the rate of fully
depolarizing noise in elementary two-qubit entangling gates.

\paragraph{Remark on other entanglement measures}

For two-qubit pure states, several standard entanglement measures
are equivalent up to monotone transformations.
In particular, one may formulate the conjecture using
\emph{concurrence} or \emph{entanglement of formation}.
A formulation based on concurrence would yield a slightly stronger
bound near product states, but the entropy-based formulation stated
above is sufficient for the purposes of this paper.

\subsection{Localizable entanglement}

Given a quantum state, the localizable entanglement of a pair of qubits,
as defined by Verstraete--Popp--Cirac \cite{VPC04}, is the maximum average
entanglement that can be obtained between them by performing single-qubit
measurements on all remaining qubits and conditioning on the outcomes.

The conjecture extends to such \textbf{localizable entangled pairs}: whenever a pair of qubits exhibits non-zero localizable entanglement, there exists a small probability $t>0$ (depending linearly on the value of the localizable entanglement) that \textbf{both qubits are replaced by the maximally mixed (maximum-entropy) state.}

\paragraph{Remark.}
In my papers I struggled to formulate the conjecture precisely and to identify
the most appropriate notion of entanglement. A formulation in terms of
localizable entanglement was proposed in my 2009 paper \cite{Kal09}. 
There I used the term ``emergent entanglement'' for the maximum expected
entanglement obtainable by separately measuring the other qubits, thus
combining optimization over measurements with averaging over their outcomes.

\subsection{An operational principle behind the extensions}

There is an operational way to motivate the extension from Bell states to
general entangled states, and from entangled states to localizable
entanglement. Every entangled two-qubit pure state can be converted, with
positive probability, into a Bell state by a suitable local filtering
measurement. Thus a partially entangled pair has a nonzero Bell-pair content
in this operational sense.

This motivates the following principle:

\begin{quote}
\emph{The fully depolarizing noise component predicted by FDNC cannot be
removed, or made arbitrarily small, by embedding the entangled pair into a
larger state and then obtaining it by local measurements, filtering, and
postselection.}
\end{quote}

For a two-qubit pure state in Schmidt form
\[
|\psi\rangle
=
\sqrt{\lambda}|00\rangle+\sqrt{1-\lambda}|11\rangle ,
\]
the optimal probability of converting a single copy of $|\psi\rangle$ into a
Bell pair by local filtering is
\[
2\min(\lambda,1-\lambda).
\]
This gives another natural way to measure the amount of entanglement relevant
to the conjecture. It vanishes for product states and equals $1$ for Bell
states. It is closely related, for two-qubit pure states, to the concurrence
\[
2\sqrt{\lambda(1-\lambda)}
\]
and to the entropy of entanglement used above. These quantities are monotone
functions of one another, although they have different behavior near product
states.

The same principle also motivates the extension to localizable entanglement.
If a many-qubit state can be measured, by local measurements on the other
qubits, so as to produce an entangled state on a pair $A,B$, then the
conjectured fully depolarizing component for the pair should persist for such a localization procedure. Thus the Bell-state, two-qubit entangled-state, and localizable-entanglement
forms of FDNC are different levels of the same operational principle.

\subsection{Implications for larger entangled states}

For more complicated entangled states---surface-code states, GHZ states, random-circuit sampling states, and cluster states---the extended conjecture applies to every pair of physical qubits that can exhibit localizable entanglement. 
In several canonical families, such as GHZ states and connected graph states,
this includes every pair of qubits. For quantum error-correcting states,
including surface-code states, one expects many pairs, and in some formulations
essentially every pair, to exhibit substantial localizable entanglement.

If true, this would have severe consequences for quantum error correction: correlated depolarizing noise on pairs of qubits is far more damaging than the quasi-independent noise assumed in threshold theorems. The reason is that a noise channel in which the events ``qubit $i$ is depolarized'' and ``qubit $j$ is depolarized'' are substantially positively correlated for all (or even most) pairs $(i,j)$ necessarily leads to large-scale error synchronization. We will elaborate on error synchronization in Section \ref {s:ce}.

\section{Consequences for Large Entangled States}
\label {s:ce}
In this section, we explain why the conjecture for pairs of entangled qubits is damaging for 
entangled states on more than two qubits and especially for attempts to create quantum error-correcting codes. We first describe a simple mathematical model that explains why the conjecture would lead to error synchronization, and move on to describe the effect of error synchronization on the error rate.

\subsection{A Simple Probabilistic Model: Tail Bounds from Pairwise Marginals}

To illustrate how strong pairwise correlations can force even large-scale error synchronization, consider the following simple probabilistic model. There is a probability distribution $W$ on $\{0,1\}^n$. We generate a random bitstring $w$ according to the distribution $W$. If $w_k=1$ we fully depolarize qubit $k$.

Let $X=(x_1,\dots,x_n)\in\{0,1\}^n$ be a random 0--1 vector of length $n$, and write $S=\sum_{i=1}^n x_i$. We assume that the joint distribution of every pair $(x_i,x_j)$ (for $i\neq j$) is fixed. The goal is to minimize the upper tail probability $\Pr(S\ge \lceil sn\rceil)$.
\begin{theorem}[symmetric ``one-parameter'' pair law]\label{thm:pairlaw}
Assume that for every $i\neq j$,
\[
\Pr(x_i=0,x_j=0)=1-t,\qquad \]
\[\Pr(x_i=0,x_j=1)=\Pr(x_i=1,x_j=0)=\Pr(x_i=1,x_j=1)=\frac{t}{3}.
\]


Let $S=\sum_{i=1}^n x_i$ and let $K=\lceil s n\rceil$.  
Then for $2\le K < \frac{n+3}{2}$, in the regime
\[
t \le \frac{3(K-1)}{n+2K-3},
\]
the minimum possible value of $\Pr(S\ge K)$ over all such distributions equals
\[
\min \Pr(S\ge K)
=
\frac{t(n-2K+3)}{3(n-K+1)}.
\]

\end{theorem}

In particular, if $K=\lceil s n\rceil$ with $0<s<\tfrac12$, then for large $n$,
\[
\Pr(S\ge s n)
\;\ge\;
\frac{1-2s}{3(1-s)}\,t
+o(1).
\]

Theorem~4.1 demonstrates that pairwise correlations can force a
non-negligible probability of observing a linear number of simultaneous
errors, even when the individual and joint error probabilities are small.

A proof and discussion of this implication will be given in the appendix.

\subsection{Error rate} 
\label {s:er}

There is an important distinction between error rate measured by distance from the ideal quantum state and error rate measured by the number of physical qubits affected by errors. Let $\rho$ be the ideal state and suppose that the actual state is \[ \rho'=(1-\epsilon)\rho+\epsilon\sigma . \] Then the trace distance between $\rho'$ and $\rho$ is at most of order $\epsilon$: \[ D(\rho,\rho') \leq \epsilon , \] where \[ D(\rho,\rho')=\frac12\|\rho-\rho'\|_1 . \] If the system now evolves by a noiseless unitary operator $U$, then \[ D(U\rho U^*,U\rho' U^*)=D(\rho,\rho'). \] Thus the error remains of size $\epsilon$ when measured in trace distance. However, the same error may look very different when measured by the number of physical qubits affected. A probability-$\epsilon$ error acting on one qubit contributes order $\epsilon$ to the expected number of qubit errors. But after applying a unitary computation $U$, the same error may be transported to an error acting nontrivially on many qubits. 
When the error is pushed forward through the computation, an error operator
$E$ becomes
\[
E\mapsto UEU^* .
\]
Even if $E$ is a one-qubit error, the operator $U E U^*$ may have large support. Thus a trace-distance error of size $\epsilon$ can correspond to order $\epsilon$ qubit errors in one description, but to order $n\epsilon$ qubit errors after the error has been spread by the computation. Equivalently, consider a noise process in which, with probability $\epsilon$, one qubit is hit, and with probability $1-\epsilon$ no error occurs. This has trace-distance scale $\epsilon$ and qubit-error scale $\epsilon$. But if a later unitary evolution spreads this error over $cn$ qubits, where $c>0$ is a constant, then the trace-distance scale is still $\epsilon$, while the expected number of qubit errors becomes $c n \epsilon$. This illustrates why trace distance alone does not capture the full burden of an error for fault tolerance. Trace distance measures distinguishability from the ideal state and is preserved by noiseless unitary evolution. By contrast, the number of qubits affected by an error is a dynamical quantity: it depends on how the error propagates through the computation. Error synchronization exploits precisely this distinction. A small amount of noise in trace distance may correspond to rare error events, but those rare events may involve many qubits simultaneously. Thus the same trace-distance error can correspond either to $O(\epsilon)$ qubit errors or to $O(n\epsilon)$ qubit errors, depending on the spatial structure of the error after propagation through the computation. 

\subsection{Some broader directions}
\label{s:broader}

The primary test for FDNC and its consequences is experimental; this will be discussed
 in Section~\ref{s:exp}. However, another way to examine the
conjecture is to look for extensions, consequences, and theoretical
connections beyond the specific context of quantum computers.

Two broad assertions are involved. First, noise accumulation in entangled
quantum systems cannot be suppressed in the manner required for scalable
fault-tolerant quantum computation. Second, for noisy quantum states and
evolutions, at least when the noise level is low, there may be systematic
relations between the noise and the corresponding ideal state or evolution.

Over the years, I made several attempts to place these ideas into a broader
mathematical framework, extending beyond the specific, and hypothetical,
setting of quantum computers. While the words ``noise'' and ``error'' are
especially natural in the engineering language of quantum devices, the same
issues are closely related to the more general notions of approximation,
perturbation, stability, and effective description in theoretical physics.

One possible principle is that noise should respect, or at least reflect, the
symmetries of the ideal state or evolution. In this form, the question is
naturally connected with representation theory, covariant quantum channels,
and the commutant of the symmetry representation associated with the
``signal.'' This direction is related to Prediction 11 in \cite{Kal16}. At
present I regard it as a useful organizing idea rather than a precise
formulation of FDNC.

A second direction is the study of restricted classes of noisy continuous-time
evolutions, which I called \emph{smoothed Lindblad evolutions}. The guiding
idea is that realistic noise should not be modeled as an arbitrary
instantaneous perturbation, but should be constrained by a short time-window
of the ideal evolution. This was one of the motivations behind Prediction 6,
``Time smoothing,'' in \cite{Kal16}.

A third direction concerns intrinsic time-parametrization. For time-dependent
quantum evolutions, I proposed that the rate of unavoidable noise should be
bounded below by a measure of the noncommutativity of the ideal evolution
during the relevant time interval; this was Prediction 7 in \cite{Kal16}.
This idea resembles the unsharpness principle
appearing in Polterovich's work on symplectic geometry of quantum noise
\cite{Pol14}. It was also inspired by conversations with R. Kosloff concerning
specific open quantum systems.


\section{Experimentally testing the conjecture}
\label{s:exp}

The Fully Depolarizing Noise Conjecture is, in principle, experimentally testable. One of the reasons for
formulating it in terms of physical qubits and physical noise channels is that
it leads to predictions that differ sharply from those of standard local noise
models. The difficulty is that the conjecture concerns not merely the total
amount of noise, but the structure of the noise channel and, in particular,
the presence of correlated fully depolarizing components.

\subsection{Experiments with a handful of qubits}

A first class of tests involves only a small number of qubits. Already for
three or four qubits one can compare different ways of producing an entangled
pair of physical qubits. For example, one may prepare a Bell pair directly by
a two-qubit gate, or indirectly by creating a small GHZ-type state and then
measuring the intermediate qubits. In standard gate-local noise models, the
genuine joint fully depolarizing component in the effective channel on the
final pair may be quadratic in the elementary two-qubit error rate when the
entanglement is produced indirectly. Here ``joint fully depolarizing
component'' refers to a channel-level component in which both final qubits are
depolarized together, not merely to the loss of fidelity of the final Bell
state. FDNC predicts that such quadratic suppression should not occur: once
the two physical qubits are entangled, an order-$t$ joint fully depolarizing
component should be present.

A direct way to test FDNC is to characterize the effective two-qubit gate
rather than a single entangled output. For the mediated construction,
retain both intermediate measurement outcomes with the appropriate
corrections and keep the mediator initialization fixed. Process tomography
then consists of preparing a tomographically complete family of known
endpoint inputs, applying the same physical gate to each, and reconstructing
the corresponding output density matrices.

Entangled inputs are not required. For example, choose each endpoint
independently from
\[
|0\rangle,\quad |1\rangle,\quad |+\rangle,\quad
|+i\rangle=\frac{|0\rangle+i|1\rangle}{\sqrt2}.
\]
The density matrices of these sixteen product inputs span the
two-qubit operator space.
Together with suitable output measurements, they allow reconstruction of
the effective channel and estimation of its largest joint fully depolarizing
weight. One can then compare this weight for direct and indirect
implementations. Preparation and measurement errors must be accounted for,
with uncertainties reported, and the physical gate implementation must
remain fixed across input preparations. Its entangling action should also
be verified.

\paragraph{Remark.}
For an isolated Bell state, full depolarization of one qubit and joint
depolarization both produce \(I_{AB}/4\). Non-maximally entangled pure
states distinguish these particular channels, but other noise channels
can still agree on their outputs. Thus state tomography of a single
entangled output, or even a restricted family of such outputs, need not
identify the joint fully depolarizing component. Process tomography
addresses this ambiguity by probing the same channel on a spanning
family of inputs.

\subsection {Experiments for larger entangled states}

A second class of tests concerns larger entangled states, say on 10--50
physical qubits. Natural candidates include GHZ states, cluster states,
surface-code states, random-circuit-sampling states, and states arising in
quantum error-correcting codes. In these systems FDNC predicts not only
pairwise correlated depolarization, but also error synchronization. Small
trace-distance noise may correspond to rare events in which many physical
qubits are affected simultaneously. Equivalently, the error rate measured in
terms of qubit errors may scale up dramatically, even when the trace-distance
weight of the bad event remains small.

One way to express the predicted signature is to consider the random set of
qubits hit by an error event in one round. Standard local noise models predict
that, if the single-qubit error rate is $p$, then for most pairs
\[
\Pr(i,j\hbox{ are both hit})\approx p^2.
\]
FDNC predicts instead that for many pairs with substantial localizable
entanglement
\[
\Pr(i,j\hbox{ are both hit})\ge cp
\]
for some constant $c>0$. Thus the experimental signature is not merely a
larger total error rate, but a different correlation structure among error
events. In particular, one should look for error synchronization: unexpectedly
large tails in the number of qubits, checks, or detection events affected in a
single round.

Quantum error-correcting codes provide a particularly natural testing ground.
The syndrome data obtained in repeated rounds already contain information
about spatial and temporal correlations among errors. One may therefore look
for excess correlations between detection events, unusually heavy syndrome
patterns, or tails of the distribution of the number of affected checks that
are incompatible with standard local stochastic noise. Such tests would not by
themselves prove the conjecture, but they could reveal whether the noise
behaves in the synchronized way predicted by it.

Random-circuit-sampling experiments provide another possible testing ground,
but also illustrate a difficulty. Fidelity measures such as XEB primarily
estimate the probability that no error occurred. They are much less sensitive
to the internal structure of the error conditioned on an error occurring. Thus
testing FDNC in such systems would require going beyond global fidelity
estimates and examining lower-order marginals, correlations between qubits,
correlations between Pauli observables, or other statistical signatures of
synchronized errors.

\subsection{Inferring the channel and experimental avenues}

There is a fundamental difficulty in all these experiments: the noise channel
is not directly observed. It must be inferred from data, usually through some
form of tomography, randomized measurements, gate-set tomography, or statistical
model fitting. This inference is complicated by state-preparation and
measurement errors, calibration errors, drift, leakage, crosstalk, coherent
errors, and the gauge freedoms inherent in gate-set tomography and related reconstruction methods. Moreover, the
conjectured noise may not appear as a clean additional mechanism separate from
familiar sources of noise. It may well arise partly, or even largely, from mechanisms already known to experimentalists and regarded as engineering difficulties. The point of the conjecture is that such difficulties may not be
merely engineering obstacles, but could reflect an inherent obstruction to
maintaining highly entangled physical states with sufficiently benign noise.

Several experimental avenues are natural. Superconducting devices, including
cloud-accessible IBM quantum computers, are especially relevant for initial
small-qubit tests because they allow flexible circuit-level experiments and
broad access. Trapped-ion devices, with high-fidelity gates and often
all-to-all connectivity, may be useful for controlled tests of indirect
entanglement generation. Neutral-atom systems are especially relevant for
larger entangled states, many-qubit correlations, and emerging quantum
error-correction experiments. Each platform has different strengths, and the
conjecture may manifest itself differently across them.

A realistic experimental program could therefore proceed in stages. First,
perform small-qubit experiments designed to distinguish linear from quadratic
scaling of the joint depolarizing component for indirectly created entangled
pairs. Second, perform tomography or randomized-measurement tests on
non-maximally entangled pairs, where single-qubit and two-qubit depolarization
are easier to distinguish. Third, analyze syndrome data from quantum
error-correcting experiments for signs of error synchronization. Fourth,
reexamine random-circuit-sampling data not only through fidelity or XEB, but
through statistical tests sensitive to the correlation structure of the errors.

No single experiment can prove the conjecture. At best, experiments can
support it, constrain the parameter regime in which it may hold, or refute
specific formulations of it. But even partial tests would be valuable. The
central issue is whether the noise in highly entangled physical states is
merely small, or whether it is small in a way compatible with the assumptions
needed for scalable fault-tolerant quantum computation.

\section{Noise Sensitivity and the Argument Against Quantum Supremacy}
\label{s:nscc}

Beginning around 2013, I pursued a different skeptical direction regarding
quantum computation. Unlike the Fully Depolarizing Noise Conjecture discussed
earlier, this line of work does not postulate exotic forms of noise. Rather,
it relies on standard noise assumptions and on ideas from computational
complexity and the theory of noise sensitivity. Its main conclusion is that
constant-level noise may destroy precisely the high-degree structure of the
ideal sampling distributions on which the complexity-theoretic hardness
arguments rely.

The starting point was joint work with Guy Kindler \cite{KalKin14}, motivated
by Boson Sampling and by the proposal of Aaronson and Arkhipov \cite{AarArk13}.
Boson Sampling offered a route to demonstrating ``quantum (computational) supremacy''
without the need for universal fault-tolerant quantum computers. Our analysis
suggested, however, that the output distributions of noisy Boson Sampling
devices are highly noise sensitive and that, at realistic noise levels, the
resulting noisy distributions admit low-degree approximations. Here
``low-degree'' refers to the Fourier--Hermite or Fourier--Walsh expansion of
the relevant probability distribution, truncated to terms of bounded degree.
Consequently, the computational complexity associated with ideal Boson Sampling
would disappear before reaching regimes where quantum supremacy could be
demonstrated.

This observation suggests that
noisy intermediate-scale quantum (NISQ) devices cannot produce computational
tasks beyond the reach of classical computers.
It also suggests that the requirements for scalable quantum error correction
are even more demanding: producing and maintaining good logical qubits appears
to require stronger control of noise than the control needed for NISQ quantum supremacy
experiments.

The central mathematical notion in our analysis is \emph{noise sensitivity},
introduced in the study of Boolean functions and percolation \cite{BKS99}. A
function is noise sensitive if a small random perturbation of its input almost
completely destroys the information carried by the output. The phenomenon is
closely related to Fourier analysis and to the concentration, or lack of
concentration, of the Fourier spectrum on low-degree terms.

The main conclusion of Kalai--Kindler \cite{KalKin14} and of several of my
subsequent papers \cite{Kal16,Kal18,Kal20} is the following.

\begin{quote}
\emph{In the asymptotic constant-noise regime, samples obtained by NISQ
devices represent efficient classical computation.}
\end{quote}

\noindent
In other words, the high-complexity features of the ideal quantum sampling
distribution are destroyed by constant noise, and the resulting noisy
distribution can be approximated by a low-complexity 
classical process.\footnote{In Kalai--Kindler \cite{KalKin14} and subsequent works, we identified
a low-level complexity class, denoted \({\bf LDP}\), contained in
\({\bf AC}^0\), the class of polynomial-size bounded-depth Boolean
circuits, which we argue well approximates the noisy samples produced
by NISQ devices.}

Related conclusions were obtained in later work by other researchers. For
example, Gao and Duan \cite{GaoDuan18} gave efficient classical simulations of
broad classes of noisy quantum computations at constant noise rate. Aharonov,
Gao, Landau, Liu, and Vazirani \cite{AGLLV22} gave a polynomial-time classical
algorithm for sampling from noisy random quantum circuits in the
anti-concentrating regime, up to inverse-polynomial total variation distance.
These and other works provide further evidence for the view that noise destroys
the computationally hard structure of NISQ sampling distributions (at least for random inputs).

This argument is conceptually distinct from the FDNC. The latter postulates a structurally damaging form of correlated
noise. The noise-sensitivity argument, by contrast, relies on conventional
models of noise and does not require additional physical assumptions. Unlike
the FDNC, which I did not regard, without
further experimental or theoretical support, as sufficient to change people's
{\em a priori} assessments regarding the possibility of quantum computing, the
argument with Kindler provides, in my opinion, a strong reason to doubt the
possibility of quantum supremacy \emph{without} quantum fault tolerance, and
in particular to doubt the computational power of NISQ devices. Indirectly, this also raises doubts about whether NISQ devices can provide a
path toward the good-quality quantum error correction required for
fault-tolerant quantum computation. 

The quantum-supremacy claims announced since 2019, beginning with the
experiment of Arute \emph{et al.} \cite{Aru+19}, present a direct challenge to this viewpoint. 


Both skeptical directions remain within the framework of quantum mechanics. 
The noise-sensitivity point of view, however, suggests a nonstandard
interpretive principle about what a complex noisy quantum experiment
represents. In the ideal circuit model, a quantum experiment is described by a
single probability distribution on outcomes. In a realistic noisy
implementation of a sufficiently complex quantum process, I propose that the
relevant object is instead a cloud of probability distributions, reflecting
small uncontrollable perturbations of the state, the device, and the
environment.

In this view, noise sensitivity is not merely a technical property of certain
ideal probability distributions. It is an inherent feature of complex quantum
evolutions. The outcomes of such an experiment are of course random, but the
probability distribution governing them is not a stable, well-defined physical
object independent of microscopic details. Rather, small changes in the noise,
calibration, or environment must lead to substantially different output
distributions. Thus the ideal noiseless distribution is not only hard to sample
from; it is physically unstable as an object of description.


\section{Statistical Analysis of Experimental Quantum Advantage Claims}
\label{s:stat}

Beginning in 2019, following Google's announcement of quantum supremacy
\cite{Aru+19}, I initiated, together with Yosi Rinott and Tomer Shoham, a
program aimed at developing statistical tools for analyzing samples produced by
quantum experiments. More recently, Carsten Voelkmann joined this effort.
Unlike the Fully Depolarizing Noise Conjecture and the noise-sensitivity
approach described in the previous section, this line of work focuses directly
on experimental data and on the statistical methodology used to support claims
of quantum computational advantage.

Our starting point was the observation that the central quantity used in
random-circuit-sampling experiments, the linear cross-entropy benchmark (XEB),
provides only limited information about the underlying output distributions. In
particular, a high XEB score does not imply that the experimental distribution
is close, in any strong sense, to the ideal quantum distribution. Indeed, our analyses indicate that the deviation of the empirical distribution from the distribution described by Google's noise model is substantial.  

This observation led to a detailed statistical analysis of the data and
methodology underlying the Google supremacy experiment and related experiments
\cite{RSK22,KRS22d,KRS23,KSV25}. We emphasized the importance of transparency,
data availability, and careful statistical validation. Among the concerns
raised in our papers were the possibility that calibration methods may lead to
``classical-computing interference (see below).'' Other concerns involved the apparently unreasonably small reported gaps between XEB scores and a priori predictions, and the absence of some crucial raw data needed to check these matters independently.


A recurring theme in this work is the distinction between \emph{verification}
and \emph{validation}. Verification concerns whether the observed data are
statistically consistent with a proposed model. Validation concerns whether the
model itself accurately captures the physical process generating the data.
Statistical consistency with a given benchmark does not by itself establish
that the benchmark is the appropriate one. This is one reason why independent
access to raw experimental data and documentation is so important.

Unlike the FDNC, which predicts a specific
structural feature of physical noise, and unlike the noise-sensitivity approach,
which concerns the computational complexity of noisy quantum systems, the
statistical approach is largely agnostic regarding the underlying physics. Its
goal is more modest: to provide rigorous statistical tools for assessing the
strength of experimental evidence and to clarify what conclusions can and
cannot be drawn from currently available data.

At the same time, statistical analysis may eventually provide a useful
framework for testing broader physical hypotheses, including the FDNC. One possible approach is to augment standard
noise models with an additional correlated two-qubit depolarizing component and
then determine, through statistical fitting, whether experimental data support
such a modification. In this way, statistical methodology may serve as a bridge
between experimental observations and conjectures concerning the structure of
physical noise.

\paragraph{Classical-computing interference.} By \emph{classical-computing interference} I mean the possibility that classical computation used in the design, calibration, selection, processing, or validation of a quantum experiment affects the final reported samples or benchmarks in a way that is not part of the intended quantum sampling process. This does not refer to the ordinary classical control required to operate a quantum device. Rather, it refers to situations in which substantial classical computation, using information about the ideal circuit or its predicted output statistics, may influence which circuits are run, how the device is calibrated, which data are retained, or how the reported benchmark is obtained. In such a case, the observed agreement with the ideal quantum distribution may partly reflect classical information inserted into the experiment, rather than being solely a consequence of the physical quantum evolution.

\paragraph{Assessment of progress.}
A central question is how to assess current progress in experimental quantum
computing. The following benchmark seems important:
\begin{quotation}
Is it currently possible to produce, without classical-computing interference,
samples of size 500K for depth-14 random circuits with 20 qubits and XEB
fidelity above 0.2?
\end{quotation}
This, and much stronger assertions, were claimed by Google already in 2019. 
(In our question we refer to similar circuits to those used by Google.) Google and other groups have made stronger and
stronger claims over the years.\footnote {This benchmark is especially appropriate for superconducting quantum computing. We note that Quantinuum \cite {DHL+25} and other companies report even lower fidelity values than those reported by Google, for the XEB values of very short samples (size around 50) averaged over thousands of circuits.} 
However, in my judgment, and based on our work,
there is no conclusive evidence for a positive answer even for this benchmark.
This benchmark should be much easier than more ambitious and widely discussed
benchmarks, such as high-quality logical qubits and gates, or convincing
demonstrations of quantum advantage.

\section{Further connections to physics}
\label{s:phys}

I do not expect a mathematical theorem proving the impossibility of quantum
computers, nor do I foresee a derivation of their impossibility from a
universally accepted physical principle. However, in a world without scalable
quantum fault tolerance, or even in restricted situations where quantum fault
tolerance does not occur, the operational scope of certain physical principles
may extend beyond what is strictly implied by their formal mathematical
statements. In this section I describe four examples.

\subsection{The time-energy uncertainty principle}

Robert Alicki, in his early critical papers on quantum computers
\cite{Ali00a,Ali00b}, relied on a version of the time-energy uncertainty
principle. The time-energy uncertainty principle has a subtle status. It was
challenged already by Aharonov and Bohm \cite{AhaBoh61}, while Aharonov,
Massar, and Popescu later established a rigorous version for the problem of
measuring or estimating a completely unknown Hamiltonian \cite{AMP02}.

Atia and Aharonov \cite{AtiAha17} gave a very different perspective. They
showed that violations of time-energy uncertainty bounds are closely connected
with Hamiltonian fast-forwarding, and that examples based on Shor's algorithm
can lead to exponentially precise measurements. From their point of view, a
naive universal form of the time-energy uncertainty principle is a
misconception.

There is, however, an important qualification. The strong violations suggested
by quantum algorithms rely on the possibility of implementing the relevant
quantum computation with sufficient accuracy. In a world without scalable
quantum fault tolerance, the operational content of the time-energy uncertainty
principle may therefore be broader than what is suggested by the ideal quantum
circuit model.

\subsection{The no-cloning principle}

The no-cloning theorem is a basic theorem of quantum information theory:
one cannot clone a completely unknown quantum state \cite{WooZur82,Die82}.
The usual theorem, however, does not forbid producing many copies of a
\emph{known} quantum state. If a state is described by a circuit, one may try
to run the circuit repeatedly and thereby produce as many copies as desired.

But for complex quantum states this conclusion relies on the physical
realizability of the circuit with sufficiently small accumulated error. If
scalable fault tolerance is unavailable, then many complex known quantum states
may not be reproducible in practice, even when they have a concise formal
description. Thus in a world without fault-tolerant quantum computation, the
no-cloning theorem may be accompanied by a stronger operational principle:
complex quantum states cannot in general be reliably reproduced, even when
their ideal descriptions are known.

This is not a mathematical strengthening of the no-cloning theorem within
standard quantum mechanics. Rather, it is an operational strengthening that
would follow from the failure of scalable fault tolerance.

\subsection{Thermodynamics}

Noiseless quantum evolution is reversible, and in the ideal circuit model one
can run complicated quantum processes forward and backward. At first sight, this could be seen as being in tension with the laws of thermodynamics. A specific
thermodynamic question is whether the physical assumptions needed for scalable
fault-tolerant quantum computation are compatible with the thermodynamic
constraints of open quantum systems.

Alicki, Lidar, and Zanardi \cite{ALZ06} argued that several assumptions entering
standard fault-tolerance theory---fast gates, a constant supply of fresh cold
ancillas, and a Markovian bath---may not be mutually consistent in light of
rigorous derivations of Markovian dynamics. Alicki developed related criticisms
in other works, arguing that thermodynamic considerations may impose inherent
limitations on quantum information processing.\footnote {A related philosophical critique of the physical meaning of the threshold
theorems, emphasizing the open-system character of fault-tolerant quantum error
correction, was given by Hagar \cite{Hag09}.}

These arguments did not convince the broader quantum-information community,
and there are substantial counterarguments to them. Still, if an independent
argument were to show that scalable quantum fault tolerance is impossible, it
would give renewed significance to the thermodynamic concerns raised by
Alicki and his collaborators.

\subsection{Non-Abelian anyons and stable topological qubits}

There is a mathematical-topological reason why the quantum statistics of
identical particles is different in two spatial dimensions than in three. In
three dimensions, the exchange group is essentially the symmetric group, leading
to the familiar bosonic and fermionic possibilities. In two dimensions, the
braid group replaces the symmetric group, allowing anyonic statistics
\cite{LeiMyr77,Wil82a,Wil82b}.

This mathematical possibility is the basis for topological quantum computation.
In particular, non-Abelian anyons could in principle store quantum information
nonlocally and perform quantum gates through braiding, thereby offering a
naturally protected route to fault-tolerant quantum computation
\cite{Wil82b,MooRea91,Kit03,Nay+08}.  

In a world where scalable quantum fault tolerance is inherently impossible,
the same principles that obstruct quantum fault tolerance may also obstruct
the physical realization of stable non-Abelian anyons suitable for quantum
computation, including proposed realizations based on Majorana zero modes.
This would not mean that the three-dimensional classification theorem
literally extends to two-dimensional physics. Rather, it would suggest that
some of its operational content may persist more broadly.

\section{Yuri and I: Three Broader Reflections}
\label {s:yuri}

\begin{flushright}
\begin{minipage}{0.70\textwidth}
\raggedleft
\includegraphics[width=\textwidth]{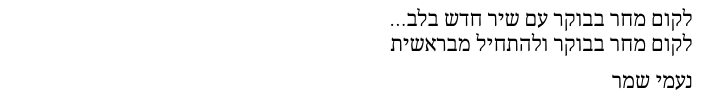}\footnotemark
\end{minipage}
\end{flushright}
\footnotetext{``To wake up tomorrow morning with a new song in one's heart... to wake up tomorrow morning and begin anew'' --Naomi Shemer}


In 1991 I met Yuri at IBM Almaden. I spent a year there, and Yuri came to give
a lecture on average-case complexity, an exciting area of computational
complexity. We had some pleasant conversations about the topic. Earlier, as a
graduate student, I knew Yuri by name and had seen him several times in
Jerusalem.

In the late 1990s I became a frequent visitor at the mathematics group at
Microsoft Research, and Yuri and I became friends. I was always fascinated by
his path from logic, through theoretical computer science, to very practical
questions about computers and algorithms, and especially by his theory of
abstract state machines \cite {Gur95,Gur00}. At that time neither of us was working on quantum
computing.

I started working on quantum computing in 2005. In 2012 I had a public debate
on the topic with Aram Harrow. I remember that Yuri criticized my hesitant and apologetic style, and his remarks gave me food for thought: “You do not have to say ‘in my opinion’ 
endlessly; everybody knows that you are expressing your opinion.'' 
Shortly afterwards, Yuri himself
joined the quantum group at Microsoft and conducted research in this direction.
Over the years Yuri grew sympathetic to my points of view, both on
why quantum computation is doomed to fail and on why Google's 2019 experimental
claims do not hold water.

\begin{figure}[]
	\begin{center}
        \includegraphics[width=13.2cm, height=9.9cm]{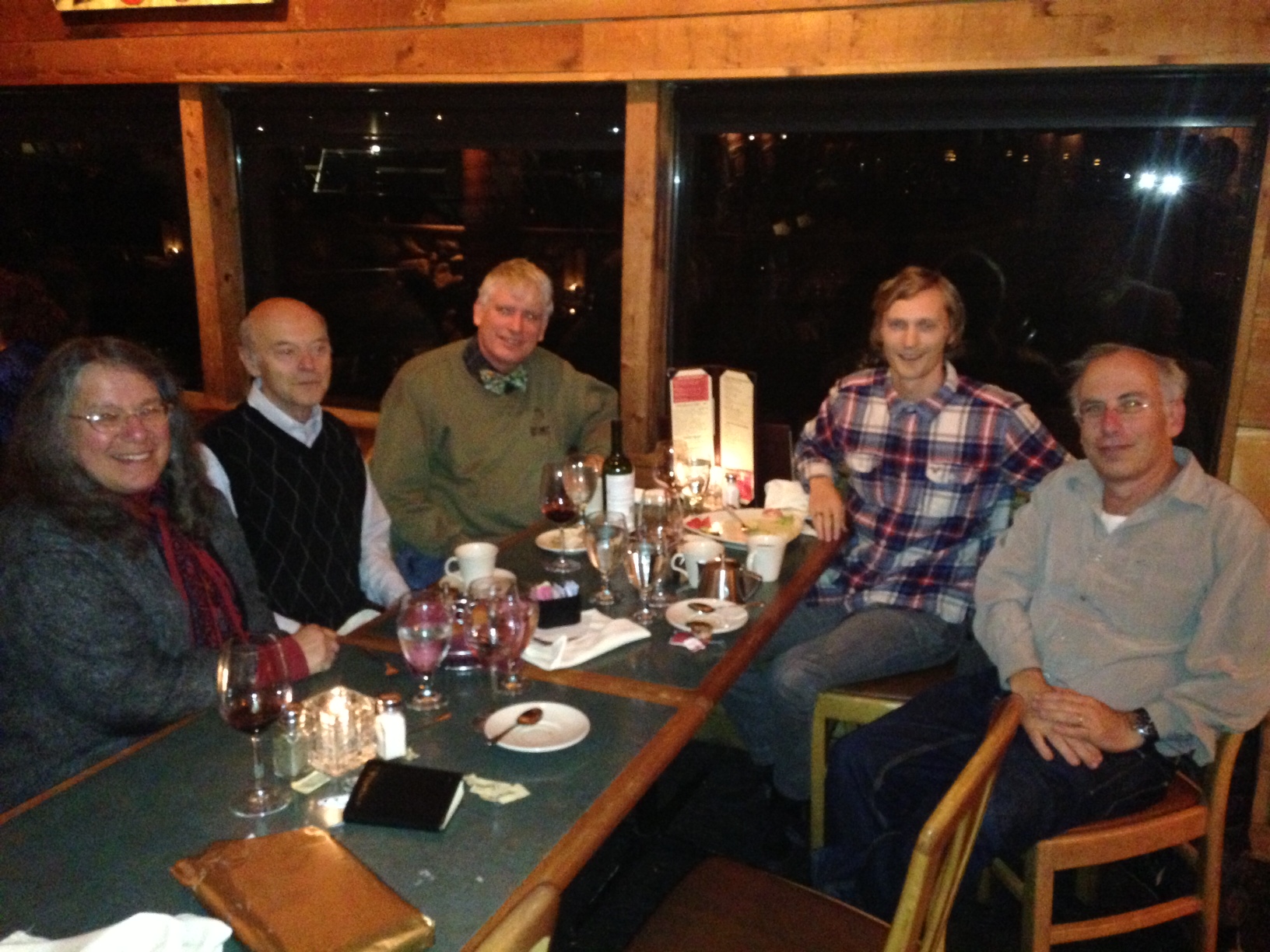}
		\caption{With Connie Sidles, Yuri Gurevich, John Sidles, and Rico Picone, in Seattle, 2013.}
		\label{fig:YS}
	\end{center}
\end{figure}

Both Yuri and I have long been interested in the relevance of theory --- especially theoretical computer science --- to practice. We have also both been
interested in the use of statistical and logical tools to detect mistakes,
implausibilities, and even deception. More recently, like almost everyone else, we have become interested in the interplay between AI on the one hand and mathematics and natural science. I will devote a few
paragraphs to each of these topics, and then end with a more personal reflection
on Israel.

\subsection{Reflection on the relevance of theory}

By the relevance problem, I mean the gap between theoretical models and the
real-world systems they are meant to describe. This issue arises when
theoretical analyses fail to capture essential features of the phenomena they
aim to explain, whether in the natural sciences, computer science, or the social
sciences. I devoted my 2018 ICM paper \cite {Kal18} to the interface between theory and
practice in three topics: linear programming, games, and quantum computers.
Let me mention here three other examples.

The computational complexity of matrix multiplication is one of the most
famous problems in theoretical computer science, and matrix multiplication is
also one of the most important practical tasks in computing. This is an area
where theory and practice interact in subtle ways: asymptotic improvements in
the exponent of matrix multiplication are among the great achievements of
complexity theory, but practical performance also depends on constants,
numerical stability, memory access, architecture, and the size regime in which
the algorithm is actually used. (See, for example, \cite{Vas12,SchZwe25}.)

Parallel computation gives another example of a persistent tension between
theory and practice (and within theory). Theoretical models such as PRAMs capture some of the
power of massive parallelism, but real parallel computation is constrained by
communication, synchronization, memory hierarchy, locality, and energy. Thus
parallel computing is a useful reminder that even very successful theoretical
models may illuminate only part of the practical computational landscape. (See, for example, \cite {VisWig85}.)\footnote{Yuri Gurevich commented: ``I think there is another, orthogonal factor though. In PRAM and other models, all actors work on the same abstraction level. This isn't so in practice. The same abstraction level is a huge simplification.''}

The third question is more philosophical. It can be expressed as follows:
are the class $\mathbf{NP}$ and the problem $\mathbf{P}\ne\mathbf{NP}$ relevant
to the actual practice of mathematicians proving theorems, and to what one
might call the ``creativity gap''?

This question has represented an ongoing debate between Avi Wigderson and me
for the last couple of decades. Avi advocates \cite {Wig06,Wig19} a strong affirmative answer and
believes that the class $\mathbf{NP}$ gives a good description of what mathematicians are trying to prove and of mathematical
proofs: proofs may be hard to find but easy to verify. I tend toward a more
negative answer. In my view, human mathematical activity is constrained by
feasible processes, and the $\mathbf{P}\ne\mathbf{NP}$ problem can at best be
regarded as a metaphor for the gap between verification and discovery. The
progress in formal proof verification gives some support to Avi's position,
since it emphasizes the distinction between finding a proof and checking one.
On the other hand, the more recent progress of AI in mathematics may be read
as lending some support to my position, by showing that aspects of mathematical
creativity may be more algorithmic than we had imagined.\footnote {Yuri Gurevich commented: ``I was there in the 1960s when people thought that decidability means feasible decidability, and undecidability means that things are certainly infeasible in practice. That proved to be naive. The super theory of the real field is decidable (by Tarski) but this buys us zilch. And relatively simple tools successfully solve the halting problem for C programs.''}

Highly optimistic interpretations of theoretical models are common in
theoretical computer science, theoretical physics, theoretical economics, and quite a few other academic
areas. Very optimistic assessment of progress are common among individuals and organizations in all areas of life. We can ask whether this level of optimism is a negative phenomenon, because it may obscure practical limitations and policymaking, or rather an instrumental tool for progress, because it pushes theory beyond what is currently achievable. 

\subsection{Reflection on the replication crisis and on catching lies with statistics}

The replication problem refers to the difficulty of reproducing published
scientific results when experiments are repeated independently. In many fields,
results that initially appeared convincing have later proved difficult, or even
impossible, to replicate.\footnote {When it comes to quantum computers, there are concerns regarding scientific claims of two technology giants. We already mentioned doubts regarding Google's ``quantum supremacy claims,'' and there are also concerns regarding Microsoft's claims about topological quantum computing.} In some cases, concerns regarding deception have also
been raised.

The replication crisis is a serious problem in several areas of science, where
scientific claims, and occasionally pseudo-scientific claims, cannot always be
trusted. The question of using statistical and mathematical tools to detect
mistakes, biases, and deceptions is both interesting and important. One source
of concern is the presence of ``too good to be true'' features in experimental
or theoretical claims. There are many famous historical examples.

This theme is closely related to Yuri's paper with Grant Olney Passmore,
\emph{Impugning Randomness, Convincingly} \cite{GurPas12}. Their paper asks how
one can convincingly challenge the alleged randomness of an event in a
real-world setting. They discuss, among other things, the role of a focal event,
the null hypothesis, and the need for the suspicious event to be specified in a
way that is sufficiently independent of the outcome. This is very close in
spirit to the statistical questions that arise when an observation looks
implausibly accurate, implausibly successful, or otherwise too good to be true.

Here is a nice exemplifying story raised by Ehud Friedgut.

\begin{quotation}
``A man claims to be able to hit a globe hanging 200 meters away with a bow and
arrow while blindfolded. An experiment is set: he shoots two arrows, a few
minutes apart, and then sends his son to fetch the globe, which is too far from
the other observers' sight. The son returns with a globe and two arrows stuck
into it as close as physically possible. This level of accuracy makes it even
harder to believe the integrity of the experiment, but we cannot yet prove our
suspicion. Now assume that we learn that while the arrows were shot, the globe
was rapidly spinning around its axis, without the father's and son's knowledge.
This means that, regardless of the father's archery skills, the longitudes of
the two arrows should be distributed uniformly. Therefore, while it is still
possible that the two arrows will end up adjacent, this would happen with
extremely low probability, and we can view their position as probabilistic
evidence that the experiment was rigged.''
\end{quotation}

This description fits rather well an argument I raised in the late 90s 
regarding the proximity
of the $p$-values of two ``Bible code'' experiments. It also fits the situation
regarding the surprisingly close agreement between some {\em a priori}  fidelity
predictions and experimental findings in Google's 2019 quantum supremacy
experiment. In both cases, the issue is not merely whether a certain event has
small probability, but whether the relevant small-probability event was
identified in a way that makes the probabilistic argument convincing.

It is often argued that the recent AI tools may worsen the replication crisis and may lead to poor-quality research. We note that, on the other hand, AI can be instrumental in performing a variety of routine, thankless, delicate, and yet important tasks, including the careful scrutiny of previous experimental and theoretical research claims. AI can be useful for scrutinizing experimental claims that depend on many unknown experimental facts, for routinely collecting raw data from authors, for checking computations reported in papers, for performing necessary
statistical tests, and for reporting the findings. In the context of fault-tolerance mathematical theorems, mathematical verification platforms can be instrumental in formal validation of crucial mathematical claims, with particular attention to the constants and to the precise assumptions on the noise model. 



\begin{figure}[]
	\begin{center}
    \includegraphics[width=0.8\textwidth]{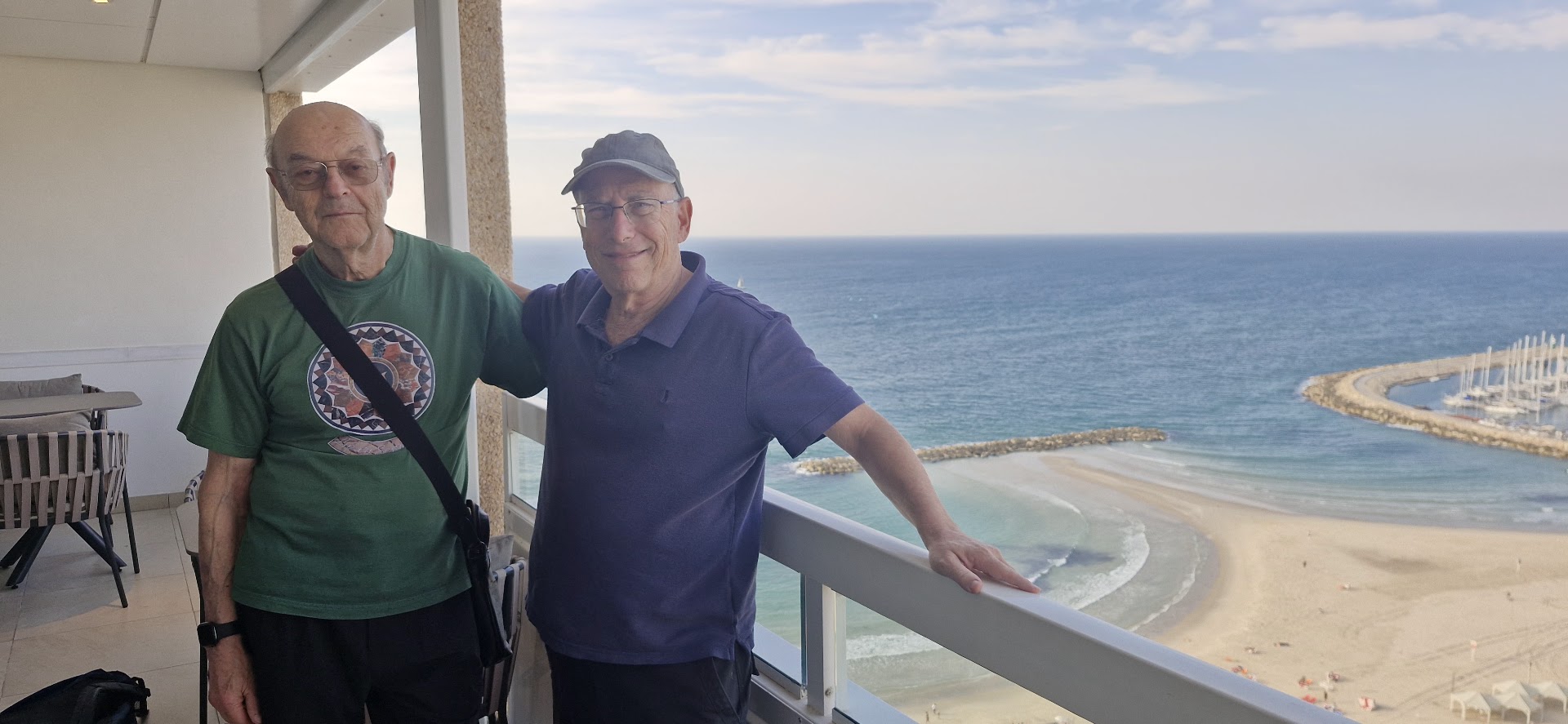}
		\caption{Yuri and I near the beach of Tel Aviv, fall 2025.}
		\label{fig:YG}
	\end{center}
\end{figure}

\begin{figure}[]
	\begin{center}
        \includegraphics[width=0.55\textwidth]{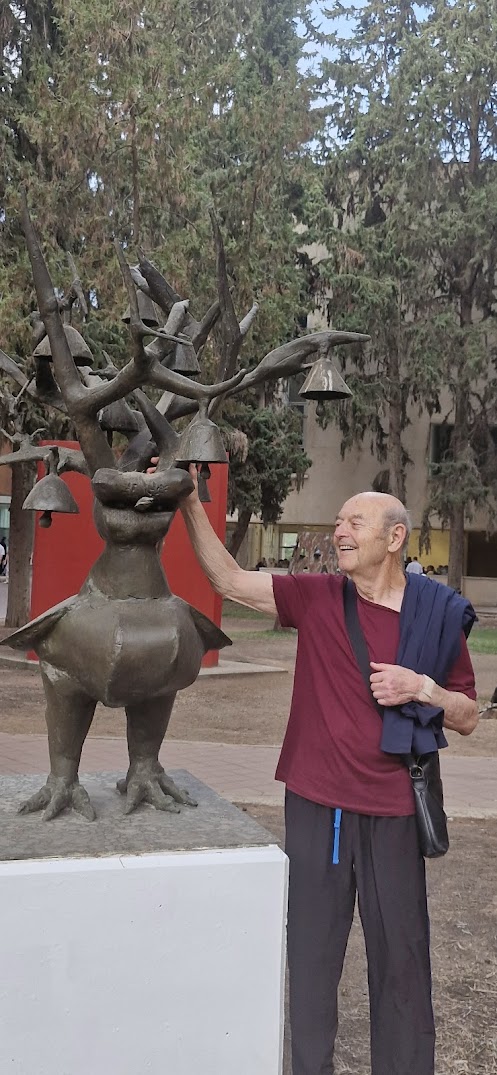}
		\caption{Yuri at the Reichman University sculpture garden, fall 2025.}
		\label{fig:YGsg}
	\end{center}
\end{figure}

\subsection{Reflection on Israel}

Finally, Yuri and Zoe regard Israel as their home, and Yuri and I share
a love, not without criticism, for the country.

Some time ago Yuri and Zoe sent me holiday greetings with a picture of
beautiful beaches. I assumed they were from the US, where they live. ``Those are
beautiful beaches,'' I wrote Yuri in response, ``and they even resemble a little
our beaches here in Tel Aviv.'' ``Gil,'' Yuri responded, ``the picture is from
Tel Aviv.''

\FloatBarrier
\section{Conclusion}
\label {s:conc}

The Fully Depolarizing Noise Conjecture was proposed twenty years ago as a structural stress test for quantum computing---a condition that scalable quantum devices must overcome. It does not attempt to describe a specific microscopic mechanism. Rather, it imposes a constraint on the effective noise channel: whenever physical qubits can generate entanglement---or localizable entanglement---correlated fully depolarizing noise must appear at linear order.

Whether this structural constraint reflects a fundamental limitation of quantum devices, or whether it will ultimately be refuted by experiment, remains an open question. The answer lies mainly in precise experimental scrutiny together with careful theoretical modeling and analysis.

My later work on quantum computers was based on computational-complexity
arguments that suggest theoretical limits to engineering efforts to reduce the
rate of noise. In recent years, I developed statistical tools to examine models for quantum noise and to scrutinize experimental claims. 

In the last thirty years, my academic and personal path crossed Yuri's path and we found many areas of common interest. I greatly cherish our friendship.

\section*{Acknowledgements}

I am grateful to Craig Gidney for several crucial discussions, in particular
for using the Stim simulator to exhibit a four-qubit example showing that FDNC
is incompatible with standard noise models in which qubit and gate errors are
statistically independent, and for pointing out the subtlety that, for an
isolated Bell state, single-qubit full depolarization and joint two-qubit full
depolarization give the same output density matrix. I also thank Sam Cohen, Michael
Geller, Yuri Gurevich, and Carsten Voelkmann for helpful discussions, comments,
and corrections. I used ChatGPT (OpenAI) as an AI-assisted tool in preparing this paper.

\section *{Appendices}

\appendix

\section{Some counterarguments and responses}
\label{s:counter}

\subsection{Is the conjecture mathematically and physically well defined?}
\label{s:counter1}

\subsection*{1) ``The conjecture does not describe a legitimate quantum channel.''}

The conjecture does not specify a complete noise model. Rather, it imposes a
constraint on whatever the correct physical noise model is. When two physical
qubits are entangled, or have positive localizable entanglement, the conjecture
asserts that the effective noise channel contains a joint fully depolarizing
component acting on this pair.

Even if universally true, the physical mechanisms leading to the conjectured
behavior may differ between implementations. The conjecture is structural, not
mechanistic.

\subsection*{2) ``The conjecture is too vague; it does not explicitly describe the noise channel. It also does not describe the physical source of the noise and its exact modeling.''}

This is partially true. The conjecture does not give a complete microscopic
description of the noise. It does not identify the bath, the Hamiltonian
couplings, leakage mechanisms, crosstalk, or other physical sources.

Rather, it proposes a structural constraint on the effective noise channel.
Testing the conjecture experimentally would require identifying in experimental
data specific correlated fully depolarizing components. Supporting it
theoretically would require fine-grained modeling of concrete physical systems.

\subsection*{3) ``As long as \(t>0\) is unspecified, the conjecture might
always remain open.''}

The intended assertion is not merely that some arbitrarily small positive
\(t\) exists. I conjecture that the joint fully depolarizing weight is
comparable to the corresponding elementary two-qubit noise scale.

The contrast with standard fault-tolerant noise models can be much sharper
than the direct-versus-indirect example of Section~2.2. Suppose, heuristically,
that a pair of physical qubits interacted \(S\) error-correction rounds ago
and that this interaction created a correlated error component of order
\(t\). In the standard fault-tolerant picture, subsequent correction can
suppress the residual contribution of this old fault exponentially with
\(S\); schematically one may expect a contribution of the form
$
t(1-c)^S,$
for some $c>0$.
(This formula is meant to illustrate the standard suppression mechanism, not
as a universal consequence of every threshold theorem.)
Since a circuit layer contains only \(O(n)\) two-qubit gates whereas there
are \(\Theta(n^2)\) pairs of physical qubits, most pairs have not interacted
directly for many rounds. More precisely, only \(O(nS)\) pairs can have
interacted during the preceding \(S\) rounds. Thus, on the standard picture,
for most pairs an error correlation inherited from their last direct
interaction can become extremely small. 

FDNC predicts a very different behavior. If such a pair presently has
substantial entanglement or localizable entanglement, the conjecture asserts
a present joint fully depolarizing component of order \(t\), rather than one
that remembers only a distant direct interaction and has decayed to
\(t(1-c)^S\). This is one reason the conjecture would be difficult to
reconcile with scalable fault tolerance.

\subsection {Relation to quantum mechanics and causality}

\subsection*{4) ``The conjecture violates linearity of QM. It is possible that it will apply to one initial state of your quantum computer but not to another one.''}

This is incorrect, although the point is interesting. As I wrote above, the conjecture proposes {\it constraints} on the noise channel rather than a complete description. 
If the same physical circuit is run with two different initial conditions, and the preparation of these initial conditions makes no physical difference to the noise, then the same noise channel should apply in both cases. In that situation, if the conjecture forces correlated depolarizing errors for one initial condition, the same effective channel also acts on the other one. There is no violation of linearity.

\subsection*{5) ``The noise (or Nature) cannot `know' if the state is entangled or not. Entanglement cannot cause correlations for the noise.''}

The conjecture does not assume that noise ``detects'' entanglement or that entanglement directly ``causes'' correlation. It asserts that the physical processes required to generate entanglement inevitably produce correlated noise.

\subsection*{6) ``The conjectured noise resembles nonphysical random-unitary models.''}

This is an important point. The need to consider noise models which resemble the behavior of random unitary operators (or, in other words, to allow inaccuracies of general form) was suggested by early skeptical views. (For example, by Leonid Levin \cite {Lev03} and John Baez.)  

Quite a few researchers argued that it is unreasonable to assume that errors for general quantum states achieved by quantum computers are limited to very restricted error operations that respect the product structure of the quantum computer.   
Proponents of quantum computers argued that while ``random unitary'' errors do not violate the postulates of quantum mechanics, they are  {\it unphysical}: they violate how quantum physics is believed to be mapped into quantum mechanics. Proponents also argue that for some nice quantum states we certainly cannot expect such noise behavior to hold. 

Based on this early discussion, the possibility of a systematic relation between the noiseless state and the noise was raised and discussed in my 2005 paper \cite {Kal05} (my first paper on the topic) and over the years has led to interesting heated discussions.

My view is the following. I accept the standard objection that an effective noise model resembling a random unitary operation on a Hilbert space of very large dimension is unphysical. But this does not by itself refute the conjecture. It may instead show that the hypothetical large-scale entangled state to which such a noise model would apply is itself physically unattainable. 

Thus the issue is somewhat counterfactual. If one assumes that a distance-15 surface-code state with the noise properties required for fault tolerance can be physically created, then FDNC may imply an effective noise channel that standard microscopic noise modeling would regard as unphysical. My conclusion from this is not that FDNC is wrong, but rather that the assumed state, with the required fault-tolerant noise behavior, is not physically available. Figure~\ref{fig:fig9}, taken from \cite{Kal11}, was meant to illustrate precisely this point.

\begin{figure}
	\begin{center}
\includegraphics[width=0.55\textwidth]{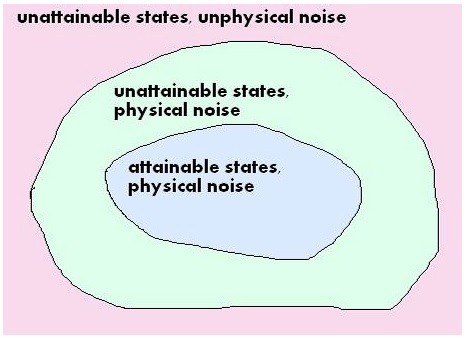}
		\caption{Given a proposed architecture for a quantum computer, it is possible that for some hypothetical states that cannot be achieved, the proposed properties of noise are ``unphysical.'' The place to examine the conjectures is for attainable states.}
		\label{fig:fig9}
	\end{center}
\end{figure}

\subsection {Relations to quantum fault tolerance} 

\subsection*{7) ``This may be true for physical qubits, but it is false for logical qubits.''}

Yes --- the conjecture refers to \textbf{physical qubits}.

If the conjecture holds, it lowers the prospects for achieving reliable logical qubits. In the standard fault-tolerance picture, logical qubits are protected by using highly entangled states of many physical qubits while suppressing the effective logical error rate. FDNC predicts instead that the underlying physical pairs with substantial localizable entanglement carry correlated fully depolarizing components. Thus sufficiently good logical qubits would be hard to reconcile with FDNC at the physical level.

\subsection*{8) ``The threshold theorem extends to general long-range noise.''}

Yes --- but these extensions still violate the conjecture. Mathematically speaking, many small long-range interaction terms are not the same as imposing substantial correlated depolarization.

\subsection*{9) ``If the error rate is $p$, the conjecture implies $\Omega(n^2p)$ errors per round. (We expect $\Omega(np)$ errors per round.)''}

This intuition captures an important feature of the conjecture.

An elementary fault occurring with probability \(t\) and initially
acting on one qubit contributes \(t\) to the expected number of affected
qubits. If subsequent evolution spreads this fault over \(cn\) qubits,
where \(c>0\) is a constant, its accumulated contribution becomes \(cnt\).
The fault probability remains \(t\), and noiseless unitary propagation
preserves the corresponding trace-distance deviation, which is at most
\(t\). Thus the same injected noise can produce a factor-\(n\) increase
in its accumulated qubit-error burden. See Section \ref {s:er}.

\subsection*{10) ``This kind of noise is likely to `prove too much' by also destroying classical fault tolerance.''}

As stated, the conjecture has no direct bearing on classical computers.

\subsection*{11) ``The correlation conjecture (and the earlier line of research from 2005) has no direct bearing on topological quantum computing.''}

Right. However, the arguments extend indirectly also to the proposal of topological quantum computing. In a world where scalable quantum fault tolerance is inherently impossible, the same principles that obstruct quantum fault tolerance may also obstruct the physical realization of stable non-Abelian anyons suitable for quantum computation.

\subsection {Experimental objections}

\subsection*{12) ``Experiments show that if a circuit has \(T\) gates, each
with accuracy \(p\), then the overall circuit accuracy falls approximately
as \(p^T\), just as it would for independent errors. Doesn't this show that
the correlated noise predicted by the conjecture is absent?''}

No.

In the independent-error model, \(p^T\) is the probability that none of
the \(T\) gates suffers an error. Even if the measured circuit accuracy
correctly estimates this probability, it does not determine the structure
of the noise conditional on an error \emph{having occurred}. Thus observing
a \(p^T\)-type decay is compatible with the correlated errors
addressed by my conjectures.

\subsection*{13) ``I doubt you'll get much interest in doing an experiment specifically for this purpose, because it's seen as not being worth the time due to being so unlikely. You're probably limited to running things on IBM's quantum cloud for yourself or to analyze statistics from existing experiments for yourself.''}

This is a fair assessment. 

\subsection*{14) ``There is the major complicating factor that there are correlated mechanisms in these machines, like cosmic ray hits, and at small scales it's unclear how you'd separate those from some legitimately dangerous-to-fault-tolerant-quantum-computation correlated error.''}

My conjecture is not separate from known correlated mechanisms, and it may well
be explained partially, or even largely, by known mechanisms that are regarded as
engineering difficulties. (This does not apply to cosmic rays.) What I claim is
that these seemingly engineering obstacles may reflect an inherent obstacle,
and the challenge was to formulate this inherent obstacle mathematically, as
simply as possible.

\subsection {Classical simulations for realistic quantum evolutions}

\subsection*{15) ``To show that quantum computational advantage is impossible,
it is not enough to demonstrate that realistic models of noise do not permit
quantum fault tolerance. One must also explain how realistic quantum evolutions
can be simulated efficiently on classical computers.''}

This is a fair point. For many decades there have been efforts, with notable
successes, to find classical simulations for realistic quantum evolutions.
It is also reasonable to expect that AI tools may be useful for finding
implicit or heuristic classical descriptions of realistic quantum evolutions
beyond the current state of the art. I agree that the impossibility of quantum
advantage would go hand in hand with an optimistic view regarding the scope
and success of classical algorithms for realistic noisy quantum systems.

In my view, identifying physically plausible noise models that rule out quantum
fault tolerance, and obtaining experimental support for such models, would
already go a long way toward demonstrating that scalable quantum computation is
impossible. Such efforts may also be relevant to classical simulation. In
particular, for noisy Boson Sampling, the skeptical arguments of
Kalai--Kindler \cite{KalKin14} led to efficient classical approximation
algorithms in the relevant noisy regime.

Extending the scope of FDNC and of the broader signal--noise principles
discussed in Section~\ref{s:broader} may also be useful in this direction. If
realistic noise is constrained by the structure of the ideal evolution, then
the effective noisy evolution may have lower computational complexity than the
ideal evolution.

\section{A Probabilistic Model for Error Synchronization}
\label{s:moment}

In this appendix we analyze a simple probabilistic model that illustrates
how pairwise correlations can force a substantial probability of simultaneous
errors. The model is not intended to describe a complete noise process. Rather,
it isolates one feature of the Fully Depolarizing Noise Conjecture: pairwise
joint error probabilities of order $t$ can imply synchronized error events
involving a linear number of qubits.

\subsection{Setup}

Let
\[
X=(x_1,\dots,x_n)\in\{0,1\}^n
\]
be a random vector, where $x_i=1$ means that qubit $i$ is hit by an error.
We impose the following symmetric pair law: for every $i\neq j$,
\begin{equation}\label{eq:pairlaw}
\begin{aligned}
\Pr(x_i=0,x_j=0)&=1-t,\\
\Pr(x_i=0,x_j=1)
=\Pr(x_i=1,x_j=0)
=\Pr(x_i=1,x_j=1)&=\frac{t}{3}.
\end{aligned}
\end{equation}
Here $0\le t\le \min(1, 3/4 + 3/(4n))$.\footnote {Such a random vector does not exist for all \(t\in[0,1]\); positive semidefiniteness of the covariance matrix already implies
that 
$t\le \frac34+\frac{3}{4n}.$
This issue has no effect on the low-noise regime considered below.} Thus each individual qubit is hit with probability
\[
\Pr(x_i=1)=\frac{2t}{3},
\]
while each pair is jointly hit with probability
\[
\Pr(x_i=1,x_j=1)=\frac{t}{3}.
\]
For small $t$, this is much larger than what one would obtain from independent
errors with the same one-qubit marginal. Indeed, independence would give
\[
\Pr(x_i=1,x_j=1)=\left(\frac{2t}{3}\right)^2=\frac{4t^2}{9},
\]
whereas \eqref{eq:pairlaw} gives a joint probability of order $t$.

Let
\[
S=\sum_{i=1}^n x_i
\]
denote the total number of qubits hit by errors.

\subsection{Moment identities}

The pair law \eqref{eq:pairlaw} determines the first two factorial moments of
$S$. Namely,
\[
\mathbb E S=\frac{2tn}{3},
\]
and
\[
\mathbb E[S(S-1)]
=
\sum_{i\neq j}\Pr(x_i=1,x_j=1)
=
\frac{n(n-1)t}{3}.
\]
Thus the expected number of errors is of order $nt$, while the expected number
of ordered pairs of errors is of order $n^2t$. Equivalently,
\[
\mathbb E\binom{S}{2}
=
\binom n2 \frac{t}{3}.
\]

The next theorem shows that these two moment identities already force a
probability of order $t$ for a linear-size error event.
\begin{theoremrestated}
Let $X=(x_1,\dots,x_n)\in\{0,1\}^n$ be a random vector satisfying
\eqref{eq:pairlaw} for every $i\neq j$.
Let $S=\sum_i x_i$. For an integer $K$ satisfying
\[
2\le K<\frac{n+3}{2},
\]
define
\[
m_{n,t}(K)=\min \Pr(S\ge K),
\]
where the minimum is over all distributions satisfying the pair constraints.
If
\[
t\le \frac{3(K-1)}{n+2K-3},
\]
then
\[
m_{n,t}(K)=
\frac{t(n-2K+3)}{3(n-K+1)}.
\]
\end{theoremrestated}

\begin{proof}
The pair law gives
\[
\mathbb E S=\frac{2tn}{3},\qquad
\mathbb E[S(S-1)]=\frac{n(n-1)t}{3}.
\]

Consider
\[
q(s)=\frac{s(s-K+1)}{n(n-K+1)}.
\]
For every integer $s$ with $0\le s<K$ we have $q(s)\le 0$. For
$K\le s\le n$, the polynomial $q(s)$ is nonnegative and increasing, and
$q(n)=1$, so
\[
0\le q(s)\le 1.
\]
Hence
\[
\mathbf 1_{\{s\ge K\}}\ge q(s)
\]
for every integer $s=0,\dots,n$. Therefore
\[
\Pr(S\ge K)\ge \mathbb E q(S).
\]
Now
\[
\mathbb E q(S)
=
\frac{\mathbb E[S(S-1)]-(K-2)\mathbb E S}{n(n-K+1)}
=
\frac{t(n-2K+3)}{3(n-K+1)}.
\]

It remains to show sharpness. Define a distribution for $S$ supported on
$\{0,K-1,n\}$ by
\[
\Pr(S=n)=
\frac{t(n-2K+3)}{3(n-K+1)},
\]
\[
\Pr(S=K-1)=
\frac{nt(n-1)}{3(K-1)(n-K+1)},
\]
and
\[
\Pr(S=0)=1-\Pr(S=K-1)-\Pr(S=n).
\]
Since $2\le K<\frac{n+3}{2}$, the first two probabilities are nonnegative.
The condition
\[
t\le \frac{3(K-1)}{n+2K-3}
\]
is exactly the condition that $\Pr(S=0)\ge 0$.

Given $S=s$, choose uniformly an $s$-subset of $[n]$ and set those coordinates
equal to $1$. This produces an exchangeable distribution on $\{0,1\}^n$.

For this distribution,
\[
\mathbb E S=\frac{2tn}{3},\qquad
\mathbb E[S(S-1)]=\frac{n(n-1)t}{3}.
\]
Since the distribution is exchangeable, it follows that
\[
\Pr(x_i=1)=\frac{\mathbb E S}{n}=\frac{2t}{3}
\]
and
\[
\Pr(x_i=1,x_j=1)
=
\frac{\mathbb E[S(S-1)]}{n(n-1)}
=
\frac{t}{3}.
\]
Therefore
\[
\Pr(x_i=1,x_j=0)=\Pr(x_i=0,x_j=1)=\frac{t}{3},
\]
and
\[
\Pr(x_i=0,x_j=0)=1-t.
\]
Thus the required pair law holds.

Since the only support point at least $K$ is $n$,
\[
\Pr(S\ge K)=\Pr(S=n)
=
\frac{t(n-2K+3)}{3(n-K+1)}.
\]
This proves sharpness.
\end{proof}

\begin{corollary}
Let $K=\lceil sn\rceil$ with $0<s<1/2$. If $t$ is in the low-noise range of
Theorem~\ref{thm:pairlaw}, then as $n\to\infty$,
\[
\Pr(S\ge sn)\ge
\frac{1-2s}{3(1-s)}\,t+o(1).
\]
\end{corollary}

\paragraph{Remark.}
The proof is an elementary instance of the linear programming method for
moment problems. We optimize the tail probability $\Pr(S\ge K)$ under
constraints on the first two factorial moments of $S$. The polynomial
\[
q(s)=\frac{s(s-K+1)}{n(n-K+1)}
\]
is a dual feasible certificate: it is dominated by the indicator
$\mathbf 1_{\{s\ge K\}}$ on $\{0,1,\dots,n\}$, while its expectation is
determined by the prescribed moments. The distribution supported on
$\{0,K-1,n\}$ is the matching primal extremizer. This is a finite-dimensional
version of the classical moment problem and of optimal probability
inequalities. For background on the classical moment problem, see Shohat and
Tamarkin \cite{ShoTam43} and Landau \cite{Lan87}; for the modern
convex-optimization viewpoint on sharp probability inequalities from moment
information, see Bertsimas and Popescu \cite{BerPop05}.

\subsection{Interpretation}

The corollary shows that even the pairwise information in \eqref{eq:pairlaw}
can force a substantial probability of observing a linear number of errors.

In particular, for every fixed $s<1/2$, the probability of seeing at least
$sn$ errors is bounded below by a constant multiple of $t$, in the low-noise
range of Theorem~\ref{thm:pairlaw}. Thus the error events cannot behave as if
they were nearly independent. Under independence with the same one-qubit
marginal, one would have
\[
\Pr(x_i=1,x_j=1)
=
\Pr(x_i=1)\Pr(x_j=1)
=
\frac{4t^2}{9},
\]
whereas the pair law gives
\[
\Pr(x_i=1,x_j=1)=\frac{t}{3}.
\]
Thus, for small $t$, the joint error probability is larger by a factor of
order $1/t$.

Moreover, if $2t/3<s$, then under independent errors with the same one-qubit
marginal the probability of the event $S\ge sn$ would be exponentially small
in $n$. By contrast, the pair law \eqref{eq:pairlaw} forces this event to have
probability of order $t$. This is the sense in which the pairwise correlations
force an error-synchronized component: with probability of order $t$, a linear
number of qubits may be hit simultaneously.

The trace-distance weight of such a bad event may be only of order $t$, but
the number of physical qubits affected in that event can be of order $n$.

\section{More background and related work} 
\label {s:back}

\subsection{On quantum advantage, or quantum supremacy}
\label{s:qa}

The term \emph{quantum supremacy}, now often replaced by the more neutral term
\emph{quantum advantage}, refers to a demonstration that a quantum device can
perform a computational task that is infeasible for classical computers. The
task need not be practically useful. The point is rather to exhibit a clear
separation between what can be done by a controlled quantum system and what
can be achieved by the best available classical computation.

There are two related but distinct meanings of this idea. The first is
asymptotic and theoretical: one gives a family of computational tasks and
argues, usually under complexity-theoretic assumptions, that no efficient
classical algorithm can perform the task, while an ideal quantum computer can.
The second is experimental and finite: one builds a particular quantum device,
runs a particular task, and argues that producing comparable samples or outputs
would be beyond the reach of existing or foreseeable classical computers.

The original and most famous theoretical example is Shor's algorithm for
factoring integers and computing discrete logarithms in polynomial time on a
quantum computer \cite {Sho94}. This is a compelling example of potential quantum
computational advantage, but it requires a large-scale fault-tolerant quantum
computer. It is therefore not a near-term demonstration.

Sampling tasks offer an even stronger form of quantum advantage from the viewpoint of computational complexity. Aaronson and Arkhipov \cite {AarArk13} and 
Bremner, Jozsa, and Shepherd \cite {BJS11} showed that the computational hardness of certain sampling tasks goes beyond the polynomial hierarchy ({\bf PH}). (In other words, solving such problems would be hard even for a classical computer equipped with a {\bf PH } oracle.)
Combining these results with the earlier result of Terhal and DiVincenzo \cite {TerDiV04}
shows that sampling tasks achieved by bounded depth quantum circuits already go beyond {\bf PH}.\footnote {Here we assume that {\bf PH} does not collapse.} 

Near-term proposals for quantum advantage using NISQ computers 
usually concern sampling tasks
rather than decision or search problems. Aaronson and Arkhipov's
proposal was based on sampling from
the output distribution of noninteracting photons in a linear-optical network.
Bremner, Jozsa, and Shepherd studied restricted commuting quantum circuits, often called IQP circuits. Random circuit sampling became a third central proposal: the task is to sample from the output distribution of a randomly chosen quantum circuit, with hardness arguments based on anti-concentration and average-case complexity assumptions.

Several experimental claims of quantum advantage have been made. In 2019,
Google reported \cite {Aru+19} a random-circuit-sampling experiment on its 53-qubit Sycamore
superconducting processor, claiming that the device performed a sampling task
far beyond the reach of classical supercomputers at the time. In 2020, the
USTC group reported \cite {Zho20} the photonic device Jiuzhang, based on Gaussian
BosonSampling, as another demonstration of quantum computational advantage.
Subsequent claims include superconducting random-circuit-sampling experiments
with the Zuchongzhi processors and further photonic experiments.

These experimental claims are necessarily more delicate than the theoretical
formulations. They depend not only on the performance of the quantum device,
but also on estimates of the best possible classical simulation algorithms, on
the statistical validation of the samples, and on assumptions about the noise.
Moreover, improvements in classical algorithms can change the interpretation
of an experiment after the fact.\footnote {The quantum advantage claims of the original papers regarding the Sycamore and Jiuzhang devices were largely refuted due to better classical algorithms, and in the case of the BosonSampling experiment a better classical algorithm is already offered by Kalai--Kindler's paper \cite {KalKin14}. However, the same groups and other groups offered later much stronger statements of quantum advantage that address also those better classical algorithms.}

\subsection {On correlated noise and fault tolerance}

Quantum error correction and quantum fault tolerance were partly motivated by
early skeptical arguments of Landauer, Unruh, and others concerning the
stability of quantum information. The discovery of quantum error-correcting codes and the proof of the threshold theorem \cite {AhaBen97,KLZ98,Kit97} were largely regarded as offering a good response to the issue of noise.  
After the proof of the threshold theorem, there were several views expressing concerns regarding correlated errors and several attempts to extend
fault-tolerance results to models with correlated errors. Let me mention, for example, Haroche and Raimond \cite {HarRai96} and Levin \cite {Lev03}. Preskill opined in 1998 \cite {Pre98} that NISQ\footnote {The term ``NISQ'' itself was coined by Preskill much later.} computers would be especially useful to study correlated noise, and this is very much also the theme of the present paper. 

Let me mention a few
developments that are especially close to my conjectures.

\paragraph {Reversible quantum computing, and quantum refrigerators} 
An early result about noisy reversible computation was proved in 1996 
by Aharonov, Ben-Or, Impagliazzo, and Nisan \cite{ABIN96}. They showed that reversible quantum computation (under standard noise models) reduces to log-depth quantum computation. 
Ben-Or, Gottesman, and Hassidim \cite {BGH13} considered fault-tolerant quantum computation in the context where there are no fresh ancilla qubits available during the computation, and where the noise is due to a general quantum channel. They identify noise channels for which polynomial time computation and even exponential time computation are possible.

\paragraph{Small error rates of arbitrary type.}
One observation that I made early on is that if the error rate measured in
terms of qubit-error counts is sufficiently low, then quantum fault-tolerance
methods can still allow logarithmic-depth quantum computation, even for quite
general types of errors. The argument is indicated in Section~5.3 of
\cite{Kal06}. The point is that, for logarithmic-depth circuits, an initially
local error can spread only within a limited light cone. Errors that defeat the
fault-tolerance scheme are therefore errors of large weight, and if such
large-weight errors occur with sufficiently small probability, logarithmic-depth
quantum computation can still be protected. 

This illustrates a recurring
theme: sufficiently weak or sufficiently restricted noise may allow (on paper) meaningful
forms of computation, even when fully scalable fault-tolerant quantum
computation is not available.  

\paragraph{Preskill's 2012 paper.}
John Preskill's paper \cite{Pre13}, partially triggered by my conjectures and
by earlier discussions we had on correlated errors, presented general sufficient
conditions for the threshold theorem to hold in the presence of correlated
Hamiltonian noise. These conditions require the terms in the noise Hamiltonian
that act collectively on many qubits, or on widely separated qubits, to decay
sufficiently rapidly.

These conditions are incompatible with the Fully Depolarizing Noise Conjecture.
Mathematically speaking, many small long-range interaction terms are not the
same as a substantial correlated fully depolarizing component on entangled
physical qubits. From my point of view, the very rapid decay required for
large-weight correlated errors is precisely the physically questionable part
of such general noise models.

\paragraph{Kalai--Kuperberg.}
One offshoot of my conjectures, particularly the notion of smoothed Lindblad
evolutions (see Section \ref {s:broader}), together with long email discussions with Greg Kuperberg, led to our joint paper \cite{KalKup15} \emph{Contagious error sources would need time travel to prevent quantum computation}. We proved that fault tolerance is
possible even under very general forms of ``disease-like'' spreading noise. The
proof relied on earlier ideas of Emanuel Knill, based on teleportation and the
conversion of quantum circuits to circuits of bounded quantum depth.

Greg Kuperberg viewed this paper as a successful response to my skepticism. I did not
quite see it that way. I saw our paper as clarifying an important point: if
time-smoothing in my proposed class of Lindblad evolutions 
applies only forward
in time, then fault tolerance can still succeed. To obstruct fault tolerance,
the smoothing would have to extend to both past and future, effectively
requiring a kind of ``time-travel'' correlation structure.

\paragraph{Fault-tolerant quantum sampling.}
Quantum sampling refers to probability distributions that can be described by
quantum circuits and provides some of the strongest complexity-theoretic
evidence for quantum advantage. One observation I made in the early 2010s, and
which was mentioned in Kalai--Kindler \cite{KalKin14}, is that quantum fault
tolerance would allow such sampling tasks to be implemented with extremely
small effective noise.

Under the assumptions of the threshold theorem, an ideal sampling circuit can
be replaced by a fault-tolerant circuit in which each logical qubit is encoded
using many physical qubits and each logical operation is protected by error
correction. By increasing the overhead, the total variation distance between
the noisy output distribution and the ideal output distribution can be made
exponentially small; in particular, each individual output probability can be
approximated with exponentially small additive error. For standard schemes the
overhead is polynomial in the relevant circuit parameters and in the desired
accuracy, although depending on the architecture and on finer properties of the
scheme, a quasi-polynomial overhead may sometimes be the more natural bound.
Thus, in the ideal fault-tolerant model, quantum sampling can in principle be
made essentially noiseless.

\paragraph{Remark: Formal verification.}
Threshold theorems are mathematical theorems, and as such they can in
principle be checked without relying on unknown experimental facts. Still,
putting mathematical claims under serious scrutiny is a difficult and often
thankless task. Since billion-dollar industries and large public investments
depend on claims about fault tolerance, these claims deserve an especially high
level of mathematical scrutiny.

It would be valuable to formalize some central threshold-theorem arguments in
modern proof-assistant platforms, with particular attention to the constants
and to the precise assumptions on the noise model. Simulations are also
valuable, and to a large extent are already being carried out. But formal
verification could play a complementary role: it could clarify exactly which
mathematical assumptions are used, where constants enter, and how the
conclusions depend on the locality and correlation structure of the noise.\footnote{We note that there is a recent paper that built a Lean library for the formalization of many concepts of quantum information \cite {KTM+26}.}

\subsection{Interpretation}

What is the right interpretation of Shor's algorithm? One possible reaction to
Shor's theorem is that it points toward a marvelous future technology:
large-scale quantum computers capable of factoring integers and solving other
classically difficult problems. Another possible reaction is that Shor's theorem
exposes a tension between the ideal quantum circuit model and physical reality.
From this point of view, the theorem is not only a promise, but also a warning
sign: perhaps the model allows computational phenomena that physical systems
cannot realize with sufficient accuracy.

This was close to the reaction of Leonid Levin \cite{Lev03}. Levin compared
Shor's algorithm to Shamir's factoring algorithm in the unit-cost arithmetic
model \cite{Sha79}. In that model one can factor integers very quickly, but
only because the model allows arithmetic operations on very large integers at
unit cost. The lesson is not that factoring is easy in the physical world, but
that a mathematically natural model may contain unrealistic assumptions about
the physical cost of computation. From this perspective, Shor's theorem may be
read either as a promise of future technology or as evidence that the ideal
quantum circuit model is too strong as a physical model.

Quantum sampling raises the same interpretive question in an even sharper
form. From the viewpoint of computational complexity, sampling tasks provide a
stronger form of quantum advantage than factoring. In the exact setting, the
relevant sampling tasks are not merely believed to be hard for ordinary
classical computation; efficient classical exact sampling, even with very
powerful classical resources, would imply unexpected collapses of the
polynomial hierarchy. Thus quantum sampling appears to go beyond \({\bf PH}\)
in a way that factoring does not.

The bounded-depth case makes the question sharper still. The work of Terhal
and DiVincenzo, together with later sampling-hardness results, showed that the
complexity-theoretic power associated with quantum sampling is not confined to
long quantum computations. It already appears in bounded-depth or very shallow
quantum circuits, provided one uses the appropriate adaptive or sampling
formulations. Should one believe that even bounded-depth quantum devices can
realize computational phenomena beyond the polynomial hierarchy, or should this
strengthen the suspicion that the idealized model is missing essential physical
limitations?

Fault tolerance adds another layer to the same question. What is the right
interpretation of the threshold theorem? One interpretation is that it provides
the crucial theoretical step toward building large-scale quantum computers: if
the noise rate is sufficiently small, and if the correlations in the noise are
sufficiently controlled, then arbitrarily long quantum computations become
possible. But there is another possible interpretation. One may wonder whether
the assumptions themselves are too optimistic: whether physics really allows
sufficiently low noise rates together with the required control of correlated
errors.

A related question arises from reversible computation. It follows from the result of Aharonov, Ben-Or, Impagliazzo, and
Nisan \cite {ABIN96} that even reversible quantum computation already allows factoring with quasipolynomial resources. Should this be regarded as another reason for optimism, or rather as an indication that the noise assumptions implicit in the ideal model are too optimistic?

Moving from theory to experiments, we can ask the same question about experimental claims. Since
2019 we have seen claims that noisy quantum devices performed sampling tasks
that would require, according to various estimates, thousands of supercomputer
years, then billions of years, then \(10^{25}\) years \cite {Google25c}. Most recently, in the
Jiuzhang 4.0 photonic Gaussian Boson Sampling experiment, the reported
comparison was \(25.6\,\mu{\rm s}\) for producing a sample on the quantum
device at room temperature versus more than \(10^{42}\) years for a state-of-the-art classical tensor-network simulation on the El Capitan supercomputer \cite {Liu+25}. What is the rational interpretation of such claims?


One possible interpretation is that these experiments are early manifestations
of the extraordinary computational power predicted by quantum mechanics and
complexity theory. Another possible interpretation is that the comparison
itself has become detached from the physical and statistical realities of noisy
experiments. The larger and more spectacular the claimed separations become,
the more important it is to ask whether the assumptions behind the comparison
have been independently validated.

\end {document}